\documentclass{article}
\usepackage[margin=1in]{geometry}
\usepackage{graphicx} 
\usepackage{authblk}
\usepackage{amsmath, amsfonts, amssymb, amsthm}
\usepackage{natbib}
\usepackage{mathtools}
\usepackage{enumitem}
\usepackage{xcolor}
\usepackage{url}

\mathtoolsset{showonlyrefs}
\theoremstyle{thmstyleone}

\newcommand{\indep}{\mbox{$\perp\!\!\!\perp$}} 
\newcommand{\E}{\mathsf{E}}
\newtheorem{theorem}{Theorem}

{
  \theoremstyle{definition}
  
}

\renewcommand{\E}{\mathsf{E}}
\renewcommand{\P}{\mathsf{P}}
\newcommand{\F}{\mathsf{F}}
\newcommand{\D}{\mathsf{D}}
\newcommand{\R}{\mathsf{R}}
\newcommand{\1}{\mathsf{1}}

\DeclareMathOperator*{\argmin}{arg\,min}

\title{The Counterfactual Combine: A Causal Framework for Player Evaluation}
\author[1]{Herbert P. Susmann}
\author[1]{Antonio D'Alessandro}

\affil[1]{Division of Biostatistics, Department of Population
  Health, NYU Grossman School of Medicine, USA}

\newif\ifanonymous
\anonymoustrue 

\date{}

\begin{document}

\maketitle

\abstract{
Evaluating sports players based on their performance shares core challenges with evaluating healthcare providers based on patient outcomes. Drawing on recent advances in healthcare provider profiling, we cast sports player evaluation within a rigorous causal inference framework and define a flexible class of causal player evaluation estimands. Using stochastic interventions, we compare player success rates on repeated tasks (such as field goal attempts or plate appearance) to counterfactual success rates had those same attempts been randomly reassigned to players according to prespecified reference distributions. This setup encompasses direct and indirect standardization parameters familiar from healthcare provider profiling, and we additionally propose a ``performance above random replacement'' estimand designed for interpretability in sports settings.  We develop doubly robust estimators for these evaluation metrics based on modern semiparametric statistical methods, with a focus on Targeted Minimum Loss-based Estimation, and incorporate machine learning methods to capture complex relationships driving player performance. We illustrate our framework in detailed case studies of field goal kickers in the National Football League and batters in Major League Baseball, highlighting how different causal estimands yield distinct interpretations and insights about player performance.
}


\section{Introduction}
In many sports, individual players repeatedly attempt to achieve a specific objective, such as converting a field goal in American football or reaching base in a plate appearance in baseball. Each attempt may be carried out under varying conditions affecting its difficulty, such as physical conditions, opposing player skill, and game situation. Naively evaluating players by their empirical success rate over an observed series of attempts may be biased due to selection effects: some players may face more difficult attempts than others due to systematic or random factors. Statistical methods are therefore needed to account for variation in attempt difficulty when evaluating player performance. 

The statistical aspects of sports player evaluation are closely related to problems in epidemiology and health economics. In particular, there is a close relationship with the field of healthcare provider profiling, which aims to evaluate healthcare providers (e.g., hospitals) while adjusting for differing patient case-mix \citep{iezzoni2003}. In our framework, healthcare providers play the same role as players: they are confronted with multiple patients (attempts) with varying characteristics (disease severity, comorbidities), and some providers may systematically face more difficult patient populations, necessitating risk adjustment and standardization. As such, the extensive literature on risk adjustment for healthcare provider profiling \citep{christiansen1997provider, normand1997profiling, delong1997profiling} can be adapted and applied to sports player evaluation. 

In particular, we draw on recent work framing risk adjustment within a causal inference framework \citep{daignault2017doublerobust, daignault2019mediation, susmann2025providerevaluation, mcclean2025comparingcausalparameterstreatments} to construct player evaluation metrics. We specify a causal model connecting attempt-level covariates, player identity, and attempt outcomes. We then define a general family of evaluation metrics as the average counterfactual success rate under a stochastic intervention that randomly reassigns attempts to players according to a prespecified distribution. This family encompasses interpretable metrics that compare a player's observed performance on a set of attempts to the estimated average performance had those attempts been performed by a reference population of players. The specific choice of intervention distribution determines the properties and interpretation of the metric. For example, reassigning attempts uniformly at random to any player in the league allows construction of a metric comparing a player's observed performance to a ``random replacement'' player. This type of reasoning is familiar in sports analytics through ``wins above replacement''-type measures; our approach formalizes this idea within a rigorous causal inference framework. The flexibility of this framework means metrics can be derived for different stakeholders; for instance, a team manager may prioritize metrics that reflect performance under their team's particular pattern of attempts, whereas a fan or journalist may prefer a league-wide comparison. 

Once a causal evaluation metric within our general family is chosen by specifying a particular stochastic intervention, the next step is to estimate the metric from observational data. Formally, this requires connecting the causal parameter to the observed data through identification assumptions, typically involving common support (positivity) and conditional exchangeability (no unmeasured confounders). Crucially, the strength of these assumptions depends on the chosen intervention. Some interventions, for example, necessitate a strong common support assumption requiring, in brief, that every type of attempt could in principle be taken by every player. Such an assumption may be violated if some players never encounter certain situations; for instance, a relief pitcher who only appears in late-inning, high-leverage situations will never face the same mix of batters and game states as an everyday starter. Alternative estimands may operate under weaker assumptions. Formalizing player evaluation as a causal problem makes these assumptions transparent, and allows us to leverage extant literature on, for example, sensitivity analyses and methods for addressing assumption violations \citep{petersen2012positivity, diaz2013sensitivity}. 

After establishing causal identification, we turn to statistical estimation of the evaluation metrics. We work within a nonparametric framework, imposing minimal assumptions on the data-generating distribution. Using tools from semiparametric efficiency theory \citep{bickel1993efficient, vanderVaart98}, we derive influence functions and efficiency bounds for our estimands, and construct doubly-robust, asymptotically efficient estimators. In particular, we provide an accessible review focusing on how modern techniques such as Targeted Minimum Loss-based Estimation (TMLE; \citealt{vanderLaan2006tmle, vanderLaanRose11}) can be used to obtain asymptotically linear estimators that remain valid when flexible, data-adaptive machine learning algorithms are used for nuisance-parameter estimation. The ability to incorporate machine learning makes our approach suitable for settings in which the relationships between attempt characteristics, player identity, and outcome are complex. 

We illustrate our causal player evaluation framework through two case studies. In the first, we define and estimate three metrics for field goal attempts by American football place kickers using data from National Football League (NFL) games. Field goal attempts provide an ideal setting for illustrating our approach: attempts are well-characterized by a relatively small number of observable covariates (distance, stadium type, weather, game situation), yet naive percentage-based metrics are susceptible to selection bias because different kickers face systematically different types of attempts, driven by their skills, coaching decisions, and their teams' offensive performance. In the second case study, we consider all plate appearances in a Major League Baseball (MLB) regular season, and evaluate batters based on whether they put the ball in play. Because batter success depends on a complex combination of pitcher and batter attributes and game context, flexible machine learning methods are particularly important in this setting. 

\paragraph{Prior work}
Statistical methods for player evaluation comprise a large and evolving literature, much of which is organized by sport. For instance, \citealt{kubatko2007basketball} and \citealt{terner2021basketball} review player evaluation approaches in basketball, and \citealt{pasteur2016evaluation} reviews methods used in American football. Synthesizing across sports, a common approach that is closely related to our setting uses regression models to predict outcomes conditional on player identity, game context, and attempt characteristics. For binary outcomes, logistic regression is a natural choice; for example, it has been used extensively to analyze NFL field goal success rates \citep{berry1985rating, bilder1998placekicks, berry2004, clark2013going, pasteur2014expectation}. These models provide adjusted performance measures, but typically lack an explicit causal framing and do not formalize the identification assumptions needed for counterfactual interpretation.

Causal reasoning has also been applied in sports analytics, although much of this work focuses on decision-making rather than player evaluation. Several studies analyze the causal effect of specific strategic choices, such as fourth-down decisions in American football \citep{yam2019fourthdown} or calling a timeout immediately before a field goal attempt (``icing the kicker''; \citealt{sanchez2024icing}). Other work uses causal methods to study the effects of timeouts and game momentum \citep{assis2020timeout, gibbs2022timeout, weimer2023momentum}. In contrast, examples of causal inference applied directly to player evaluation are relatively sparse \citep{terner2021basketball}. \citealt{bacco2024plusminus} describe a causal evaluation metric for individual players based on team-level outcomes, contributing to a literature that seeks to disentangle the effect of individual players on team outcomes \citep{macdonald2011plusminus, baumer2015}. Our setting differs in that we focus on metrics based on granular attempt-level outcomes, and define standardized evaluation metrics via explicit stochastic interventions. 

Machine learning has seen significant application in sports analytics \citep{koseler2017machinelearningbaseball, brefeld2022machinelearningsports, fujii2025machinelearningsports}, with an emphasis on prediction. A key advantage of our estimation approach is that it accommodates flexible, data-adaptive machine learning methods for nuisance components (such as the expected success probability conditional on covariates) while still providing valid statistical guarantees, including valid uncertainty intervals for the evaluation metrics. This allows us to leverage the predictive strength of modern machine learning algorithms while retaining rigorous inferential guarantees. 

\paragraph{Contributions}
This article formalizes player evaluation in a causal inference framework, connecting it to established literatures in epidemiology and healthcare provider profiling. We define a general class of counterfactual player evaluation metrics based on stochastic interventions that reassign attempts to players, encompassing both direct and indirect standardization estimands and ``replacement-level''-style metrics. We present identification results for these estimands, and discuss the role and plausibility of assumptions such as positivity and conditional exchangeability in sports settings. We then review concepts from semiparametric efficiency and apply them to derive doubly robust, asymptotically efficient estimators leveraging flexible machine learning methods. Finally, we illustrate our framework in case studies of NFL place kickers and MLB batters. 

\section{Causal Evaluation Metrics}
We consider \textit{attempts} as the unit of analysis. An attempt is a repeated action taken by an individual player, such as kicking a field goal in American football, making a free throw in basketball, or making a plate appearance in baseball. We consider primarily cases in which the outcome of an attempt is measured as a binary variable, e.g. whether a field goal is made, a free throw is converted, or a plate appearance results in a hit. We focus on binary outcomes, although our framework can be easily extended to non-binary outcomes as well. 

\medskip

Formally, let $Z = (X, A, Y) \in \mathcal{Z}$ be a generic random variable corresponding to an attempt, where $X \in \mathcal{X}$ is a vector of covariates characterizing the attempt (e.g. distance, location, game situation, environmental conditions), $A \in \{1, \dots, m \} \equiv \mathcal{A}$ indexes the player who made the attempt, and $Y \in \{0, 1\}$ indicates if the attempt was successful. The observed data comprise $n$ independent and identically distributed draws $Z \sim \P$, denoted $(Z_1, \dots, Z_n)$. We assume that the data-generating law $\P$ lies in the nonparametric model $\mathcal{M}$ of all probability laws on the support of $Z$. 
We use the following notation to denote elements of a law $\P \in \mathcal{M}$:
\begin{align}
    \mu_{\P}(a, X) &= \E_{\P}[Y \mid A = a, X] \\
    m_{\P}(X) &= \E_{\P}[Y \mid X] \\
    \pi_{\P}(a \mid X) &= \P(A = a \mid X).
\end{align}
The probability of treatment given covariates $\pi_{\P}(a \mid X)$ is often referred to as a \textit{propensity score}. 
As shorthand, we will write $\P(a) \equiv \P(A = a)$. 

\medskip

We formalize a causal model for player performance using the Structural Causal Models framework \citep{pearl2009}. We posit fixed functions $f_X$, $f_A$, and $f_Y$ (called \textit{structural equations}) such that
\begin{align*}
    X = f_X(U_X), \quad A = f_A(X, U_A), \quad Y = f_Y(X, A, U_Y),
\end{align*}
where $U = (U_X, U_A, U_Y)$ is exogenous. This encodes a causal model in which the attempt characteristics depend on exogenous factors, the player assigned to make the attempt depends on the characteristics of the attempt and exogenous factors, and the outcome depends on the attempt characteristics, the player, and exogenous factors.

Our proposed class of evaluation metrics are based on stochastic interventions on the mechanism assigning players to attempts. Formally, let $A^*$ be a random variable taking values in $\mathcal{A}$, drawn from a pre-specified distribution $\P_{A^* \mid X}$, which is allowed to depend on $X$. We define the counterfactual outcome under the stochastic intervention setting $A$ to $A^*$ as
\begin{align}
    Y(A^*) = f_Y(X, A^*, U_Y).
\end{align}
Intuitively, $A^*$ describes how we reassign players to attempts under a hypothetical intervention, and $Y(A^*)$ is the counterfactual outcome that would have been observed if that intervention were implemented. We target a counterfactual estimand $\psi$ as a mean counterfactual outcome under this stochastic intervention:
\begin{align}
    \label{eq:causal-estimand}
    \psi = \E\left[ Y(A^*) \mid X \in \mathcal{X}', A \in \mathcal{A}'\right],
\end{align}
where $\mathcal{X}' \subset \mathcal{X}$, $\mathcal{A}' \in \mathcal{A}$, and
we assume throughout that the conditioning set has positive probability (i.e. $\P(X \in \mathcal{X}', A \in \mathcal{A}') > 0$) so that the parameter is well defined. We discuss several specific examples below based on specific choices for the stochastic intervention and conditioning sets.  
\begin{itemize}
    \item \textbf{Direct standardization.} Fix $a \in \mathcal{A}$. Set $\P_{A^* \mid X}$ such that $A^* = a$ with probability one. This amounts to a deterministic intervention always assigning attempts to player $a$. Next, set $\mathcal{A}' = \mathcal{A}$, yielding the parameter $\psi_a^{\mathsf{direct}} = \E\left[Y(a) \mid X \in \mathcal{X}' \right]$. This parameter is referred to as a \textit{direct standardization} parameter in epidemiology and healthcare provider profiling. In causal inference, this parameter is equivalent to a mean counterfactual outcome under a categorical treatment. In our setting, this is interpretable as the counterfactual success rate had all attempts been assigned to player $a$.
    \item \textbf{Indirect standardization.} Set $\P_{A^* \mid X} = \P_{A\mid X}$, where $\P_{A \mid X}$ is the conditional distribution of $A$ given $X$. This encodes a stochastic intervention in which attempts are randomly reassigned to a player who would have been likely to make the attempt, based on its characteristics. Next, set $\mathcal{A}' = \{ a \}$, such that the reference distribution is the attempts made by player $a$. The resulting \textit{indirect standardization} parameter is $\psi_a^{\mathsf{indirect}} = \E[Y(A^*) \mid A = a]$. A player evaluation metric can be formed by the contrast $\E[Y \mid A = a] - \psi_a^\mathsf{indirect}$, which compares a player's observed performance to the counterfactual mean performance had the attempts instead been randomly reassigned to players with the same attempt profile. 
    \item \textbf{Performance above random replacement.} Set $\P_{A^*}$ such that, for all $a \in \mathcal{A}$, $\P(A^* = a) = \tfrac{1}{m}$. This encodes an intervention in which every attempt is equally likely to be taken by any player, regardless of the attempt characteristics $X$. Setting $\mathcal{A}' = \{ a \}$ yields the parameter $\psi_a^{\mathsf{rand}} = \E[Y(A^*) \mid A = a]$. For intuition, in the American Football example, this intervention could be theoretically implemented by pausing the game any time that kicker $a$ was about to make a field goal attempt and replacing them with another random kicker. An evaluation metric for player $a$ can then be formed by the contrast $\E[Y \mid A = a] - \psi_a^{\mathsf{rand}}$, which compares the observed performance of player $a$ against the average performance had they been replaced by a random alternate player. This parameter is related to the Average Treatment Effect on the Treated (ATT) parameter, extended to categorical treatments.
\end{itemize}
Additional parameters can be constructed using different choices for the conditioning sets, which we discuss next. 
\begin{itemize}
    \item \textbf{League-wide.} Setting $\mathcal{X}' = \mathcal{X}$ makes no restrictions on the types of attempts considered in the counterfactual parameter.
    \item \textbf{Team-specific}. Suppose the covariates $X$ include a variable $T \in \{ \mathsf{team1}, \mathsf{team2}, \dots \}$ indexing the team for which the attempt is to be made. For example, in the NFL case study $T$ indicates the team attempting a field goal. If we condition on $T = \mathsf{team1}$ then the causal evaluation metric will average over the population of attempts made by that team. This may be of interest to team managers who wish to compare how players would have performed if confronted by the types of attempts representative of their team.
    \item \textbf{Player-specific.} Setting $\mathcal{A}' = \{ a \}$ restricts the reference population to that of a single player. 
\end{itemize}

Some care is necessary in interpreting the proposed metrics, particularly when it comes to comparing players by their evaluation metrics. The key complication arises when the reference distribution is not the same for every player. For the direct standardization parameter $\psi_a^{\mathsf{direct}} = \E[Y(a)]$, the counterfactual outcomes for player $a$ are averaged over the same attempt distribution regardless of player. Players can therefore be straightforwardly ranked by their direct standardization metric. If $\psi_{a_1}^{\mathsf{direct}} > \psi_{a_2}^\mathsf{direct}$, then player $a_1$ would have performed better than player $a_2$ on average over the same distribution of attempts. On the other hand, indirect standardization and the random replacement parameters do not allow for straightforward comparison and ranking of players because the reference distributions are player specific. The fact that indirect standardization cannot be used to create league tables is well-known in healthcare provider profiling \citep{shahian2020misuse}. 

The causal evaluation estimands are defined in terms of counterfactual outcomes $Y(A^*)$. To connect these causal quantities with the observed data, we require several identification assumptions. 
\begin{enumerate}[label={A\theenumi}]
    \item \label{assumption:no-unmeasured-confounding} \textbf{No unmeasured confounders.} Assume $U_A \indep U_Y$ and either $U_X \indep U_A$ or $U_X \indep U_Y$. 
    \item \label{assumption:positivity} \textbf{Positivity.} For all $x \in \mathcal{X}'$ and all $a$ with $\P(A^* = a \mid X = x) > 0$, then $\pi_{\P}(a \mid x) > 0$. 
    \item \label{assumption:intervention} \textbf{Intervention independence.} The stochastic intervention $A^*$ is independent of $U_Y$ given $X$: $A^* \indep U_Y \mid X$. 
\end{enumerate}
Assumption~\ref{assumption:no-unmeasured-confounding} implies the conditional exchangeability condition that for $a \in \mathcal{A}$, $Y(a) \indep A \mid X$. Note that Assumption~\ref{assumption:intervention}, which requires that the intervention cannot reassign players in a way that incorporates information from the unobserved $U_Y$, is satisfied by construction in all of our examples. 
Assumptions~\ref{assumption:no-unmeasured-confounding}, \ref{assumption:positivity} and \ref{assumption:intervention} are sufficient to identify the causal estimand \eqref{eq:causal-estimand}:
\begin{align}
    \psi(\P) = \E_{\P}\left[ \sum_{a' \in \mathcal{A}} \E_{\P}\left[ Y \mid A = a', X \right] \P(A^* = a' \mid X) \mid X \in \mathcal{X}', A \in \mathcal{A}' \right].
\end{align}
A formal derivation is included in the appendix for completeness, although it follows standard arguments. 
Next, we discuss the assumptions and identification result for the running examples below.
\begin{enumerate}
    \item \textbf{Direct standardization.} The identification result simplifies to 
    \begin{align}
        \psi_a^\mathsf{direct} = \E[\E[Y \mid A = a, X] \mid X \in \mathcal{X}'] = \E[\mu_{\P}(a, X) \mid X \in \mathcal{X}'].
    \end{align}
    The positivity assumption \ref{assumption:positivity} requires that $\pi_{\P}(a \mid X) > 0$ holds $\P$-almost surely; in other words, the target player must have had a positive probability of having made all attempts in the reference population. In the American Football example, this assumption might be violated if there is a kicker who never attempted kicks in some circumstances. For example, there might be a player who, due to their team schedule, never attempts kicks in an open stadium in cold weather. This would violate the positivity assumption if the reference population of the direct standardization parameter includes open stadium cold weather attempts.
    \item \textbf{Indirect standardization.} The identification result simplifies to 
    \begin{align}
        \psi_a^\mathsf{indirect} = \E[\E[Y \mid X] \mid A = a, X \in \mathcal{X}'] = \E[m_{\P}(X) \mid A = a, X \in \mathcal{X}'].
    \end{align}
    Note the crucial difference between this result and the result for the direct standardization parameter: here, the inner expectation does not condition on $A = a$. This has important consequences for estimation. The positivity assumption collapses to requiring only that $\P(A = a \mid X \in \mathcal{X}') > 0$, which is a weaker positivity condition than is required for direct standardization. If $\mathcal{X}' = \mathcal{X}$, then the assumption requires only that $\P(A = a) > 0$, which is satisfied trivially.
    \item \textbf{Performance above random replacement.} The identification result simplifies to 
    \begin{align}
        \psi_a^\mathsf{rand} = \frac{1}{m} \sum_{a'=1}^m \E\left[\E[Y \mid A = a', X] \mid  A = a, X \in \mathcal{X}' \right] = \frac{1}{m} \sum_{a'=1}^m \E\left[\mu_{\P}(a', X) \mid  A = a, X \in \mathcal{X}' \right].
    \end{align}
    The positivity assumption requires that for any $X \in \mathcal{X}'$ such that $\pi_{\P}(a \mid X) > 0$, then $\pi_{\P}(a' \mid X) > 0$ for all $a' \in \mathcal{A}$. In words, for any type of attempt by player $a$, every other player must also have had a positive probability of making such an attempt.
\end{enumerate}

\section{Efficiency Theory}
In this section we discuss the statistical properties of the causal evaluation metrics, working within a nonparametric statistical model. The underlying theory related to the estimation of low-dimensional parameters in semiparametric models has a long history (see e.g. \citealt{bickel1993efficient, vanderVaart98, kosorok2008} for comprehensive treatments). We provide a brief overview, and refer to \citealt{kennedy2023semiparametric} for a deeper introduction with applications in causal inference. 

Our treatment unfolds similarly for each of the evaluation metrics. In general, suppose we wish to estimate a generic evaluation metric $\psi : \mathcal{M} \to \mathbb{R}$. First, we derive the efficiency bound for estimating $\psi(\P)$ in a nonparametric model, which involves deriving the \textit{efficient influence function} (EIF) of $\psi$. The EIF is fundamental because the variance of any regular estimator of $\psi(\P)$ is bounded below by the variance of the EIF. Second, we draw on knowledge of the EIF to construct a targeted estimator that achieves this bound. In this article, we focus on constructing such estimators using Targeted Minimum Loss-based Estimation (TMLE). 

The key to understanding the semiparametric properties of a parameter $\psi$ is to show that it is sufficiently smooth as a function of the underlying distribution so as to satisfy a type of functional first-order Taylor expansion, sometimes referred to as a \textit{von Mises expansion} \citep{robins2008higherorder, kennedy2023semiparametric}. 
Formally, this requires showing that there exists a function $\D_{\P} : \mathcal{Z} \to \mathbb{R}$ in $L_2^0(\P)$ and a function $R : \mathcal{M} \times \mathcal{M} \to \mathbb{R}$ such that, for any $\P$, $\F \in \mathcal{M}$, 
\begin{align}
    \psi(\F) - \psi(\P) = \E_{\P}\left\{ \D_{\F}(Z) \right\} + \R(\P, \F).
\end{align}
The function $\D_{\P}$ is called the (canonical) gradient of $\psi$ at $\P$, and $\R$ is called the second-order remainder term. In a nonparametric model, the canonical gradient is unique in $L_2^0(\P)$ and coincides with the \textit{efficient influence function} (EIF) of $\psi$ at $\P$.  
Establishing this expansion is crucial because the asymptotic variance of any regular, asymptotically linear estimator of $\psi({\P})$ is lower bounded by $\E_{\P}\left[ \D_{\P}(Z)^2 \right]$, the variance of the EIF. 

The EIFs for direct and indirect standardization parameters are given in \citealt{susmann2025providerevaluation}. In Appendix~\ref{appendix:semi-parametrics}, we formally establish the von Mises expansion for the ``random replacement'' parameter $\psi_a^{\mathsf{rand}}$, which includes the form of its EIF. Taken together, and fixing $\mathcal{X}' = \mathcal{X}$ for simplicity, the EIFs for the direct, indirect, and random replacement parameters specific to player $a \in \mathcal{A}$ are given by
\begin{align}
    \D^{\mathsf{direct}}_{\P}(Z) &= \textcolor{purple}{\frac{\1(A = a)}{\pi_{\P}(a \mid X)}}\left\{ Y - \mu_{\P}(a | X) \right\} + \mu_{\P}(a | X) - \psi_a^{\mathsf{direct}}, \\
    \D^{\mathsf{indirect}}_{\P}(Z) &= \textcolor{purple}{\frac{\pi_{\P}(a \mid X)}{\P(a)}}\left\{ Y -  m_{\P}(X) \right\} + \frac{\1(A = a)}{\P(a)} \left\{ m_{\P}(X) - \psi_a^{\mathsf{indirect}} \right\} , \\
    \D^{\mathsf{rand}}_{\P}(Z) &= \frac{1}{m} \sum_{a' = 1}^m \left[ \textcolor{purple}{\frac{\1(A = a')}{\P(a)}\frac{\pi_{\P}(a \mid X)}{\pi_{\P}(a' \mid X)}} \left\{ Y - \mu_{\P}(a', X) \right\} + \frac{1}{\P(a)}\left\{ \mu_{\P}(a', X) - \psi_a^{\mathsf{rand}} \right\} \right].
\end{align}
Each of the EIFs has the same basic structure: a weighted residual plus a centered conditional mean term. The key differences lie in the residual weights (highlighted in the above display), with direct implications for how difficult each parameter is to estimate. 
For the direct standardization parameter, the residual is weighted by the inverse propensity score. As a result, small propensity scores (i.e., players who rarely take certain types of attempts) inflate the variance of the EIF and can substantially increase the efficiency bound. In contrast, for the indirect standardization parameter, the residual is weighted by the propensity score itself rather than its inverse. Small propensity scores have the opposite effect from before: they down-weight rare player-attempt combinations, yielding a smaller efficiency bound. This helps to explain the popularity of indirect standardization parameters in healthcare provider profiling: the intervention distribution reflects the observed patterns in such a way that positivity violations do not affect the efficiency bound. Finally, for the random replacement metric, the weights involve ratios of propensity scores for the focal player and each other player. The efficiency bound is therefore driven by how much overlap there is between the focal player's attempt profile and other players.

\section{Estimation}
In this section we discuss methods for estimating causal evaluation metrics, with the goal of deriving estimators that achieve the efficiency bounds given in the previous section. Efficient estimators for direct and indirect parameters were proposed in \citep{susmann2025providerevaluation} in the context of healthcare provider profiling; as such, we focus in this section on discussing the random replacement parameter $\psi_{a}^{\mathsf{rand}}$. We consider here the high-level considerations when designing an estimator and include details in Appendix~\ref{appendix:tmle}. 

\subsection{Substitution Estimators}
To introduce the main considerations that arise in estimation, we start by discussing a substitution approach. Suppose we estimate an outcome regression model $\hat{\mu}$ using some regression method, for example logistic regression. Then, an estimate of the random replacement parameter can be formed by averaging the predictions from the regression model. Let $\mathcal{C} = \{ i : A_i = a, X_i \in \mathcal{X}' \}$. The substitution estimator is given by 
\begin{align}
    \hat{\psi}_a^{\mathsf{rand}} = \frac{1}{|\mathcal{C}|} \sum_{i \in \mathcal{C}} \frac{1}{m} \sum_{a' = 1}^m \hat{\mu}(a', X_i),
\end{align}
where $\hat{\mu}(a', X)$ is an estimate of $\mu_{\P}(a', X) = \E[Y \mid A = a', X]$. 
Substitution estimators are attractive because they are conceptually straightforward and easy to implement. In practice, $\hat{\mu}$ can be estimated using standard software (e.g. \texttt{glm} in \texttt{R}), and the marginalization can be performed easily with tools such as the \texttt{marginaleffects} package \citep{arelbundock2024marginaleffects}. When $\hat{\mu}$ is parametric, the contribution of each covariate to the success probability can typically be easily interpreted via the estimated regression coefficients.

The statistical properties of the substitution estimators depend entirely on how $\hat{\mu}$ is estimated. In particular, $\hat{\psi}_a^{\mathsf{rand}}$ is consistent for $\psi_a^{\mathsf{rand}}$ only if $\hat{\mu}$ consistently estimates the true outcome regression $\mu_{\P}$ (for example, in $L_2(\P)$). For example, a simple main-effects logistic regression may fail to capture important interactions or nonlinearities in the relationship between covariates and attempt success; in such cases, $\hat{\mu}$ is misspecified and the resulting substitution estimator will not be consistent.

More complex models could be chosen for $\hat{\mu}$, such as penalized regression models or machine learning algorithms like random forests or gradient boosting. Two challenges then arise. First, optimizing $\hat{\mu}$ for prediction optimizes a bias-variance trade-off for the full regression function $\mu_{\P}$. Plugging this estimate into the functional $\psi_{a}^{\mathsf{rand}}$ does not generally yield an efficient estimator of $\psi_{a}^{\mathsf{rand}}$. Second, in general we do not have a straightforward way to obtain valid inference (e.g. asymptotically valid confidence intervals) for the resulting plug-in estimate when $\hat{\mu}$ is fit with flexible methods. These limitations motivate estimators that are specifically constructed (\textit{targeted}) to be efficient for $\psi_{a}^{\mathsf{rand}}$ even when $\hat{\mu}$ is estimated with arbitrary data-adaptive methods.

\subsection{Targeted Minimum Loss-based Estimation}
The motivation behind targeted estimators can be seen by analyzing the properties of a plug-in estimator. Suppose we have an initial estimate $\hat{\P}$ of the data-generating distribution $\P$. The von Mises expansion reveals that the error of a plug-in estimator $\psi(\hat{\P})$ based on $\hat{\P}$ decomposes as
\begin{align}
    \psi(\hat{\P}) - \psi(\P) = \E_{\P}\left\{ \D_{\hat{\P}}(Z) \right\} + \R(\P, \hat{\P}),
\end{align}
where $\D_{\hat{\P}}$ is the EIF of $\psi$ at $\hat{\P}$ and $\R(\P, \hat{\P})$ is a second-order remainder term.
The first term on the right-hand side represents the first-order bias of the plug-in estimator, and the second term is the second-order bias of the plug-in estimator. To show consistency of the plug-in estimator, we would need to show that both terms converge to zero asymptotically, which is difficult to ensure in a general setting. 

One approach is to debias the plug-in estimate by estimating the first-order bias via the empirical mean of the estimated EIF and subtracting it from the initial plug-in estimate:
\begin{align}
    \hat{\psi}^{\mathsf{onestep}} \equiv \psi(\hat{\P}) - \frac{1}{n}\sum_{i=1}^n \D_{\hat{\P}}(Z_i).
\end{align}
If conditions can be found to control the second-order remainder term $\R(\P, \hat{\P})$, it is possible to show consistency and asymptotic normality. Such estimators are typically referred to as \textit{one-step estimators} because they are analogous to a one-step Newton-Raphson update, with the EIF playing the role of the gradient of the functional \citep{pfanzagl1982, bickel1993efficient, fisher2021}. Under appropriate rate conditions on the nuisance estimators, $\hat{\psi}^{\mathsf{onestep}}$ is asymptotically normal and efficient. A drawback of this approach is that  the estimator is not guaranteed to lie in the parameter space; in our case, for example, a one-step estimate of the random replacement parameter may fall outside of $[0, 1]$. 

TMLE takes a different approach by carefully \textit{fluctuating} the initial estimate $\hat{\P}$ of $\P$ along a low-dimensional parametric submodel in such a way that the updated estimate $\hat{\P}^*$ (approximately) solves the empirical EIF estimating equation; that is,
\begin{align}
    \frac{1}{n}\sum_{i=1}^n \D_{\hat{\P}^*}(Z_i) \approx 0.
\end{align}
A targeted estimate of $\psi_{a}^{\mathsf{rand}}$ is then formed by plugging $\hat{\P}^*$ into the parameter mapping: $\hat{\psi}_{a}^{\mathsf{rand}} \equiv \psi_a^{\mathsf{rand}}(\hat{\P}^*)$. As a substitution estimator, the targeted estimate is guaranteed by construction to lie in the parameter space. Constructing an appropriate fluctuation scheme that zeros the empirical first-order term is primarily a mathematical exercise; we give a complete description of the proposed algorithm in Appendix~\ref{appendix:tmle}. In terms of consistency, the structure of the EIF and remainder term implies that the resulting estimator will be consistent if \textit{either} the outcome regression $\mu_{\P}$ or the attempt assignment model $\pi_{\P}$ is estimated consistently, a result commonly referred to as \textit{double robustness}. 

By design, the targeted estimate eliminates the leading (first-order) term in the von Mises expansion. The remaining task is to establish conditions under which the second-order remainder term $\R(\P, \hat{\P}^*)$ is asymptotically negligible. In our setting, the form of the remainder can be shown to have a product form, involving the estimation errors of the outcome regression $\mu$ and the attempt assignment model $\pi$. As a result, the targeted estimator can be shown to be asymptotically normal and efficient as long as $\hat{\mu}$ and $\hat{\pi}$ converge at rates faster than $n^{-1/4}$. This \textit{rate double robustness} rate \citep{rotnitzky2020mixedbias} is the crucial finding that allows the use of data-adaptive nuisance estimation methods that converge at slower than parametric rates.

Our estimation strategy relies crucially on \textit{cross-fitting} \citep{schick1986semiparametric}. Cross-fitting proceeds by splitting the data into multiple disjoint folds. For each fold, nuisance estimators are fit using observations in the complementary folds (the training set), and out-of-sample predictions are produced for observations in the held-out (validation) fold. Iterating over folds yields a complete set of nuisance predictions, where each prediction is computed out of sample. Cross-fitting allows us to use arbitrary, potentially complex nuisance estimators, or ensembles of multiple estimators. Without cross-fitting, asymptotic normality results typically require restricting the complexity of the nuisance estimators. A classical way to achieve this is by assuming nuisance estimators fall in a Donsker class, verified through entropy arguments \citep{vanderVaart1996}. Cross-fitting obviates these requirements and is now a standard practice in semiparametric estimation \citep{zheng2011cvtmle, chernozhukov2018dml}.

We are now ready to formally state the statistical properties of the cross-fitted, targeted estimator of $\psi_{a}^{\mathsf{rand}}$.
\begin{theorem}[Consistency and asymptotic normality]
    \label{thm:consistency}
    Let $\| \cdot \|$ denote the $L_2(\P)$ norm, and $\hat{\psi}_a^{\mathsf{rand}}$ be the TMLE estimate of $\psi_{a}^{\mathsf{rand}}$. Assume there exists $\epsilon > 0$ such that $\pi(a' \mid X) > \epsilon$ for all $a' \in \mathcal{A}$, $\P$-almost surely. Suppose that either $\| \hat{\pi} - \pi \| = o_{\P}(1)$ or $\| \hat{\mu} - \mu \| = o_{\P}(1)$. Then 
    \begin{align}
        \hat{\psi}_{a}^{\mathsf{rand}} - \psi_{a}^{\mathsf{rand}} = o_{\P}(1) .
    \end{align}
    Suppose that both $\| \hat{\pi} - \pi \| = o_{\P}(n^{-1/4})$ and $\| \hat{\mu} - \mu \| = o_{\P}(n^{-1/4})$. Then 
    \begin{align}
        \sqrt{n}\left( \hat{\psi}_a^{\mathsf{rand}} - \psi_{a}^{\mathsf{rand}} \right) \leadsto N\left(0, \E_{\P}\left[ \D^{\mathsf{rand}}_{\P}(Z)^2 \right]\right).
    \end{align}
\end{theorem}
A proof sketch is given in Appendix~\ref{appendix:tmle}. 
The stronger positivity condition in the theorem, requiring that the propensity scores be bounded away from zero, is a commonly used sufficient condition for showing asymptotic normality. Based on the asymptotic normality results in Theorem~\ref{thm:consistency}, Wald-type confidence intervals for the TMLE point estimates can be formed using the empirical standard deviation of the estimated EIF as a standard error estimate. 

\section{Case Studies}
We illustrate the definition and estimation of causal evaluation metrics in case studies of NFL kickers and MLB batters. In both case studies, we use the 
\ifanonymous
[anonymized] package
\else
\texttt{TargetedRisk} package (\url{https://github.com/herbps10/TargetedRisk})
\fi
for estimation. We use ensemble learning \citep{breiman1996stacking} (specifically through the \textit{Super Learning} \citealt{vanderlaan2007superlearner} framework) with the \texttt{mlr3superlearner} package for nuisance estimation. Reproduction code for the case studies is available at \ifanonymous
    [anonymized].
\else
    \url{https://github.com/herbps10/causal_combine_paper}.
\fi

\subsection{NFL Kickers}
In our first case study, we define and estimate evaluation metrics for NFL kickers based on field goal attempts. We use data from all games played during the 2013--2023 seasons collected using the \texttt{nflreadr} package \citep{ho2025nflreadr}. The raw data include 11,523 plays marked as field goal attempts performed by 110 unique kickers. Extra point attempts following a touchdown were not included. We restricted our analysis to the 46 kickers with at least 100 attempts to ensure reasonably precise player-specific estimates, yielding an analysis dataset of 9,786 attempts. The empirical field goal success rate in this dataset is 85.3\%. 

As covariates $X$, we include the kick distance (measured in yards), an indicator for an outdoor (non-enclosed) stadium, an indicator for grass vs artificial turf, temperature (Fahrenheit), wind speed (miles per hour), and an indicator for attempts occurring in the final 30 seconds of the game. For enclosed stadiums, wind speed was not recorded; we set wind speed variable to zero in these cases, reflecting that on-field wind exposure is typically minimal.

For each kicker in the analysis dataset, we estimate direct, indirect, and random replacement evaluation metrics. For the direct metric, we compare kickers based on the targeted estimates $\widehat{\psi}_{a}^\mathsf{direct}$, which under our identification assumptions can be interpreted as the mean success rate that would have been observed if every field goal attempt had been taken by the focal kicker, under the required identification assumptions. Because every kicker is evaluated on the same reference distribution of attempts, these standardized success rates can be directly compared across kickers.
For the other parameters, we calculated, for each kicker $a$, the difference between their empirical success rate and their indirectly standardized and random replacement parameters, yielding the metrics
\begin{align}
    \label{eq:indirect-metrics}
    \widehat{\Delta}^\mathsf{indirect}_a = \widehat{\E}[Y \mid A = a] - \psi_{a}^\mathsf{indirect} \quad \text{ and } \quad  \widehat{\Delta}^\mathsf{rand}_a = \widehat{\E}[Y \mid A = a] - \psi_{a}^\mathsf{rand}.
\end{align}
Recall that under causal assumptions, $\psi_a^\mathsf{indirect}$ is interpreted as the success rate that would be observed if kicker $a$'s attempts were randomly reassigned to kickers with similar attempt profiles. The random replacement parameter has a related interpretation: $\psi_a^{\mathsf{rand}}$ is the success rate had kicker $a$'s attempts been randomly reassigned uniformly at random to any kicker in the league. In both cases, the reference distribution is player-specific. As such, kickers cannot be directly compared to one another using either metric; rather, they are useful for assessing whether a kicker performs better or worse than expected with respect to a reference distribution of players. For all metrics, the outcome regression $\mu_{\P}(a, X) = \E_{\P}(Y \mid A = a, X)$ and attempt assignment model $\pi_{\P}(a \mid X) = \P(A = a \mid X)$ are estimated using a Super Learner ensemble including the \texttt{mean}, \texttt{glm}, and \texttt{lightgbm} libraries. All estimates are constructed using 10-fold cross-fitting. 

Assigning the estimates a causal interpretation depends on the validity of each parameter's identification assumptions. All three types of parameter depend on the no-unmeasured-confounding assumption. While it is impossible to verify this empirically, our covariates include the most important factor (kick distance) and a range of physical and game-situation variables. The direct standardization and random replacement metrics additionally require positivity conditions ensuring that each kicker has a positive probability of attempting each type of field goal in the reference population. While this holds in principle, in practice some kickers are very unlikely to attempt certain kicks; for example, a coach may avoid long field goals with a kicker perceived as weak from distance. Such practical positivity violations may manifest as increased uncertainty (wider intervals) for those kickers.

We begin discussing the results with the directly standardized field goal success rates. Because kickers can be directly compared against each other using this metric, we present the results in a leaderboard format, sorting kickers by the point estimate of their standardized success rate (Figure~\ref{fig:nfl-direct-standardization}). The results rank Evan McPherson, Tyler Bass, and Justin Tucker as the top three kickers in terms of their directly standardized success rates. It's somewhat surprising, at first glance, to see Phil Dawson as the player with the lowest standardized success rate in terms of the point estimate. For fans of Dawson, we note two points: the analysis period includes only the latter part of his career (excluding his tenure from 1999-2012 at the Cleveland Browns), and his performance is statistically indistinguishable from many other kickers once uncertainty is taken into account. 

Interestingly, several kickers have 95\% confidence intervals that do not overlap their empirical success rate. For example, Tyler Bass has a higher directly standardized success rate than his observed rate, which is consistent with him having faced a more difficult mix of attempts than the league-wide reference distribution. Conversely, several players have standardized success rates lower than their empirical rates, suggesting that they either faced a less difficult mix of attempts or were particularly well matched to the situations in which they were used. 

\begin{figure}
    \centering
    \includegraphics[width=0.6\linewidth]{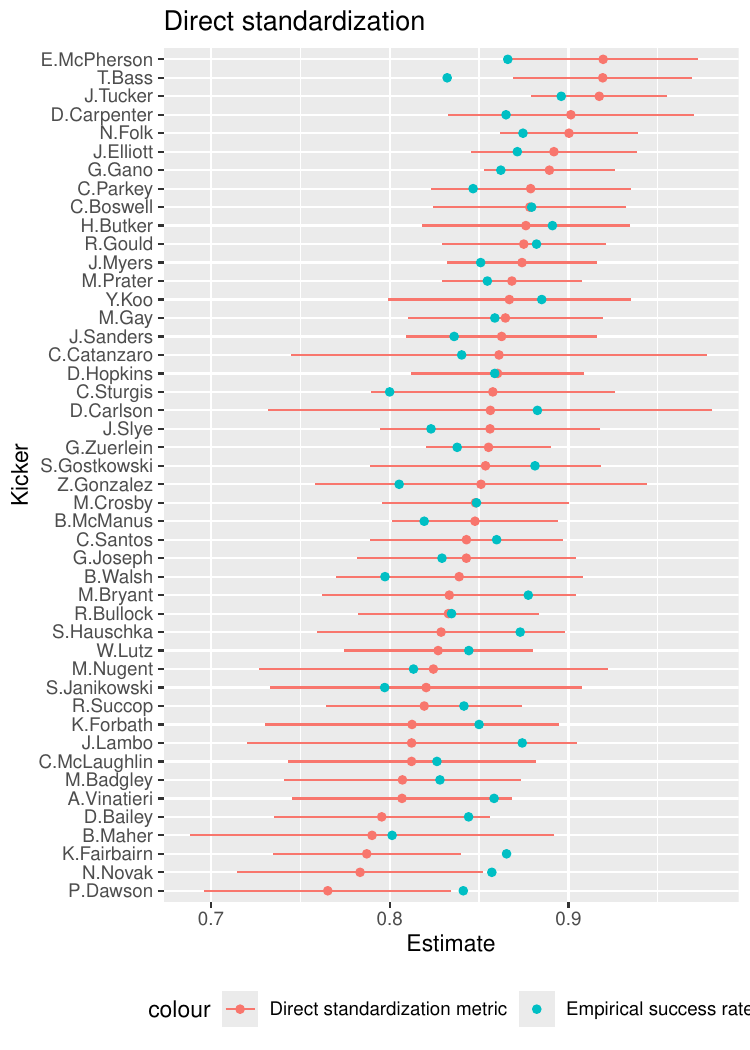}
    \caption{\small Directly standardized field goal success rates for each NFL kicker in the analysis dataset. The red points indicate the point estimate and the horizontal lines depict the 95\% confidence interval. The blue points show the empirical field goal success rates. Kickers are arranged in descending order of their point estimate.}
    \label{fig:nfl-direct-standardization}
\end{figure}

To display the indirect and random replacement metrics, we use funnel plots, as is common in healthcare provider profiling \citep{spiegelhalter2005funnel, griffen2012funnel}. Funnel plots focus attention on statistically meaningful deviations from expectation while avoiding misleading rank-based interpretations. Funnel plots show the precision of the estimate on the $x$-axis, separating estimates by their statistical certainty. The metric is placed on the $y$-axis. Points falling above or below the plotted control lines are significantly different from zero at different confidence thresholds.

Funnel plots for the indirect and random replacement metrics are shown in Figure~\ref{fig:nfl-funnels}. Interestingly, none of the NFL kickers had estimated metrics significantly below zero; this is consistent with the hard truth that poorly performing kickers are unlikely to be retained in the league. Conversely, several kickers stand out as performing above expectation. Justin Tucker has statistically significant positive values for both the indirect and random replacement metrics.

Interpreting metrics based on indirect standardization can be challenging because the reference distribution is based on the focal kicker's attempt profile. One way of thinking about this is that the comparison between two kickers is based on the similarity in the propensity scores estimated for every attempt. Players with similar propensity scores will have higher weight in the counterfactual comparisons at the heart of the indirect standardization metric. To illustrate this visually, Figure~\ref{fig:nfl-bass-comparison} compares the cross-fitted propensity scores for Tyler Bass, Stephen Gostkowski, and Phil Dawson for every field goal attempt in the analysis dataset. Bass and Gostkowski have correlated propensity scores, reflecting that they faced a similar profile of field goal attempts. In particular, Bass and Gostkowski both played for teams with home stadiums featuring artificial turf during the analysis period; as such, both players have a higher estimated propensity of attempts on artificial turf than on grass. On the other hand, Dawson played for the San Francisco 49ers and Arizona Cardinals during this period, both of which have home stadiums with natural grass fields. This partly explains why his propensity score profile differs from that of Bass. The implication of this for Tyler Bass's indirectly standardized metric is that his observed success rate is compared to a counterfactual success rate had attempts been assigned to other players, with similar players like Stephen Gostkowski being more likely to be chosen than Phil Dawson in the counterfactual assignment.

Extending this analysis, one way to visualize the similarities between players is through clustering of their estimated propensity scores. To illustrate, we construct a distance measure between kickers $a$ and $a'$ as the Euclidean distance between their normalized propensity scores:
\begin{align}
    d(a, a')^2 = \sum_{i=1}^n \left\{\bar{\pi}(a \mid X_i) - \bar{\pi}(a' \mid X_i) \right\}^2, 
\end{align}
where $\bar{\pi}(a \mid X_i) = \hat{\pi}(a \mid X_i) / \sum_{j=1}^n \hat{\pi}(a \mid X_j)$. Then, we applied a hierarchical clustering algorithm (using the \texttt{hclust} function in \texttt{R}) and visualized the result as a dendrogram (Figure~\ref{fig:nfl-cluster-dendrogram}). The result helps contextualize the indirect standardization parameters: loosely speaking, kickers close to each other in the dendrogram have more similar propensity score distributions, and are compared with higher weight in the indirect standardization results. Tying back to the Tyler Bass example, we see that he is placed closely to Stephen Gostkowski, who played in the same division, and Daniel Carpenter, who also played for the Buffalo Bills.

\begin{figure}
    \centering
    \includegraphics[width=\linewidth]{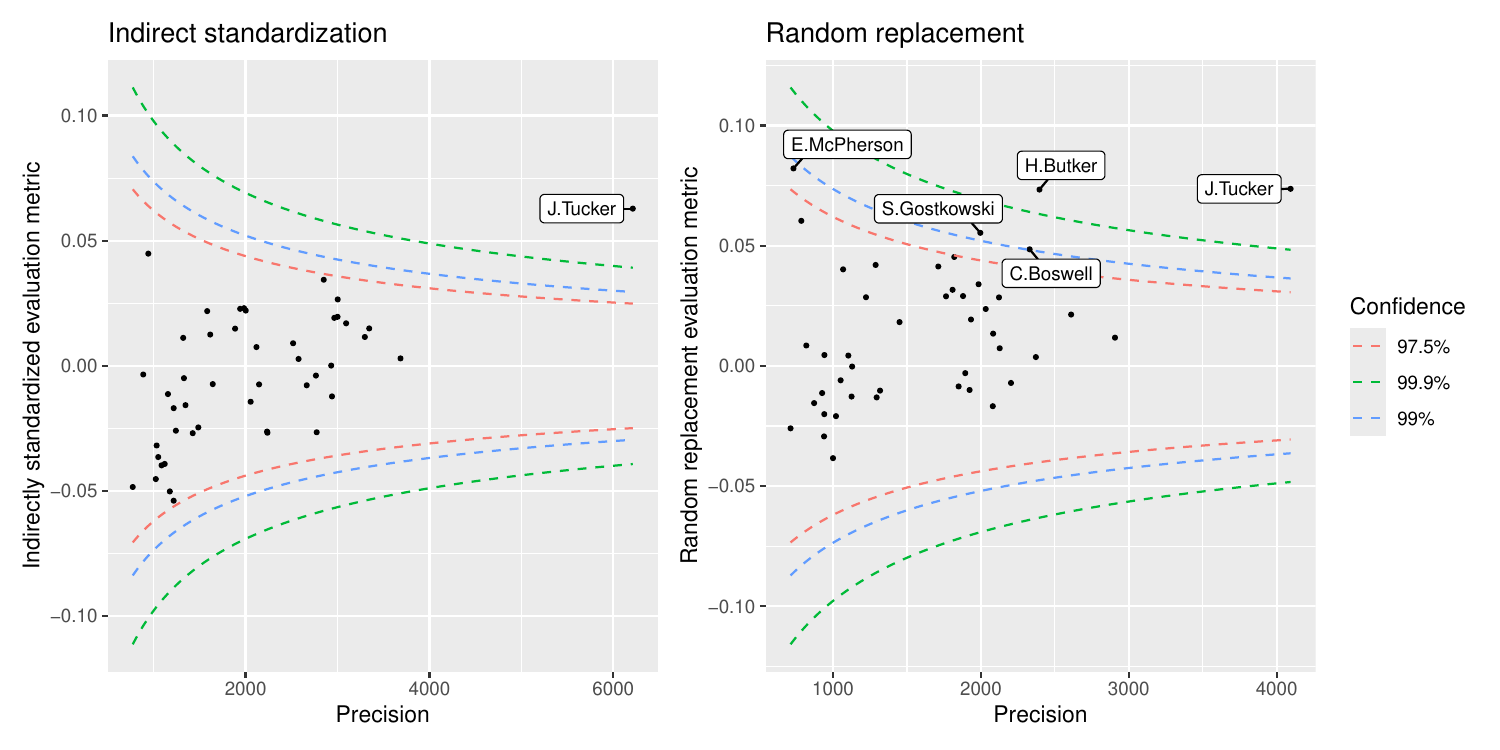} 
    \caption{\small Funnel plots visualizing the estimated indirect standardization (left) and random replacement (right) metrics \eqref{eq:indirect-metrics} for each NFL kicker in the case study. The $x$-axis shows the statistical precision of the point estimate, and the $y$-axis is the point estimate. Dashed lines indicate control limits at the 97.5\%, 99\%, and 99.9\% confidence levels. }
    \label{fig:nfl-funnels}
\end{figure}

\begin{figure}
    \centering
    \includegraphics[width=0.9\linewidth]{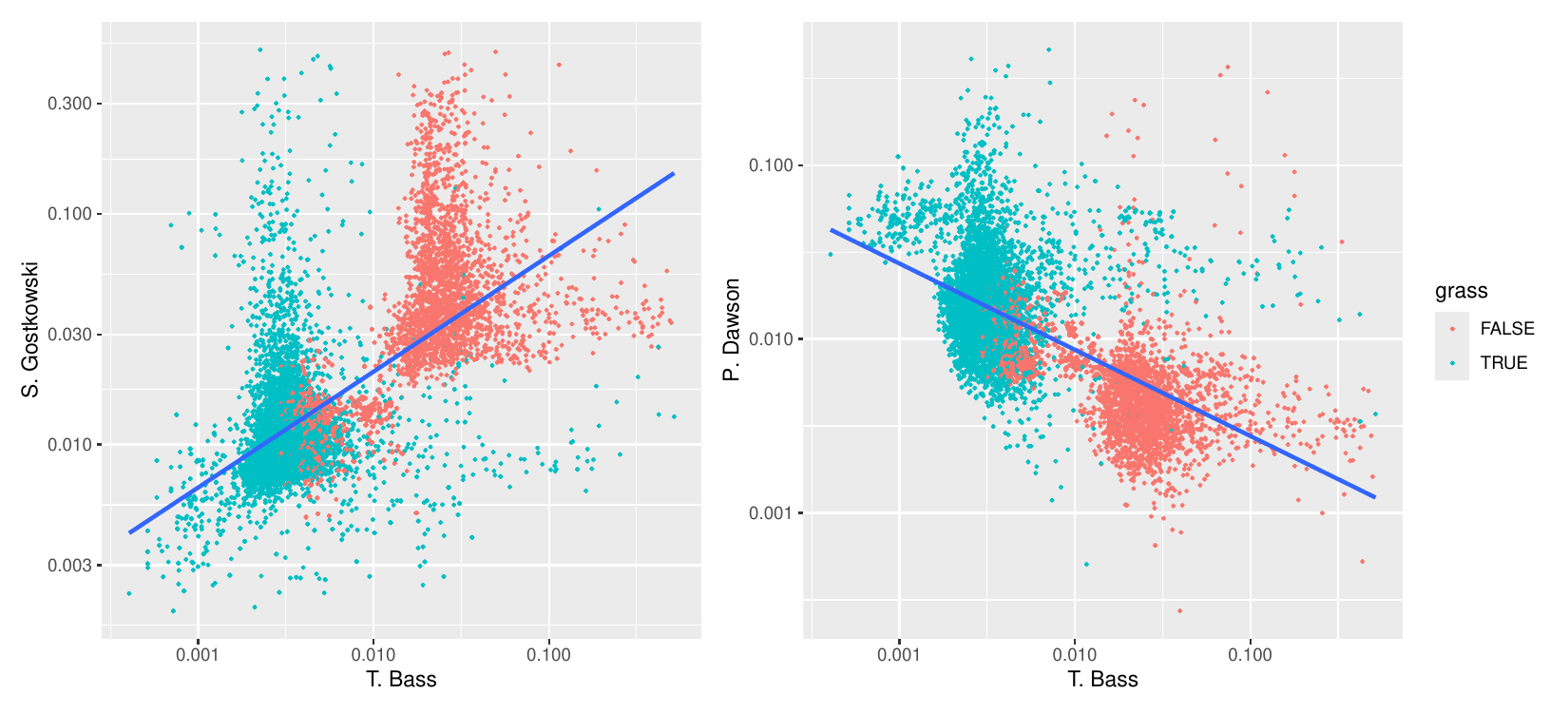}
    \caption{\small Comparison of the cross-fitted propensity score estimates for each field goal attempt in the analysis dataset for S. Gostkowski, P. Dawson, and T. Bass. Each point represents a single field goal attempt. Attempts performed on a grass surface are shown in blue, and attempts on artificial turf are shown in red. The linear best-fit line is included to emphasize the correlation between propensity score estimates for each player.}
    \label{fig:nfl-bass-comparison}
\end{figure}

\begin{figure}
    \centering
    \includegraphics[width=0.9\linewidth]{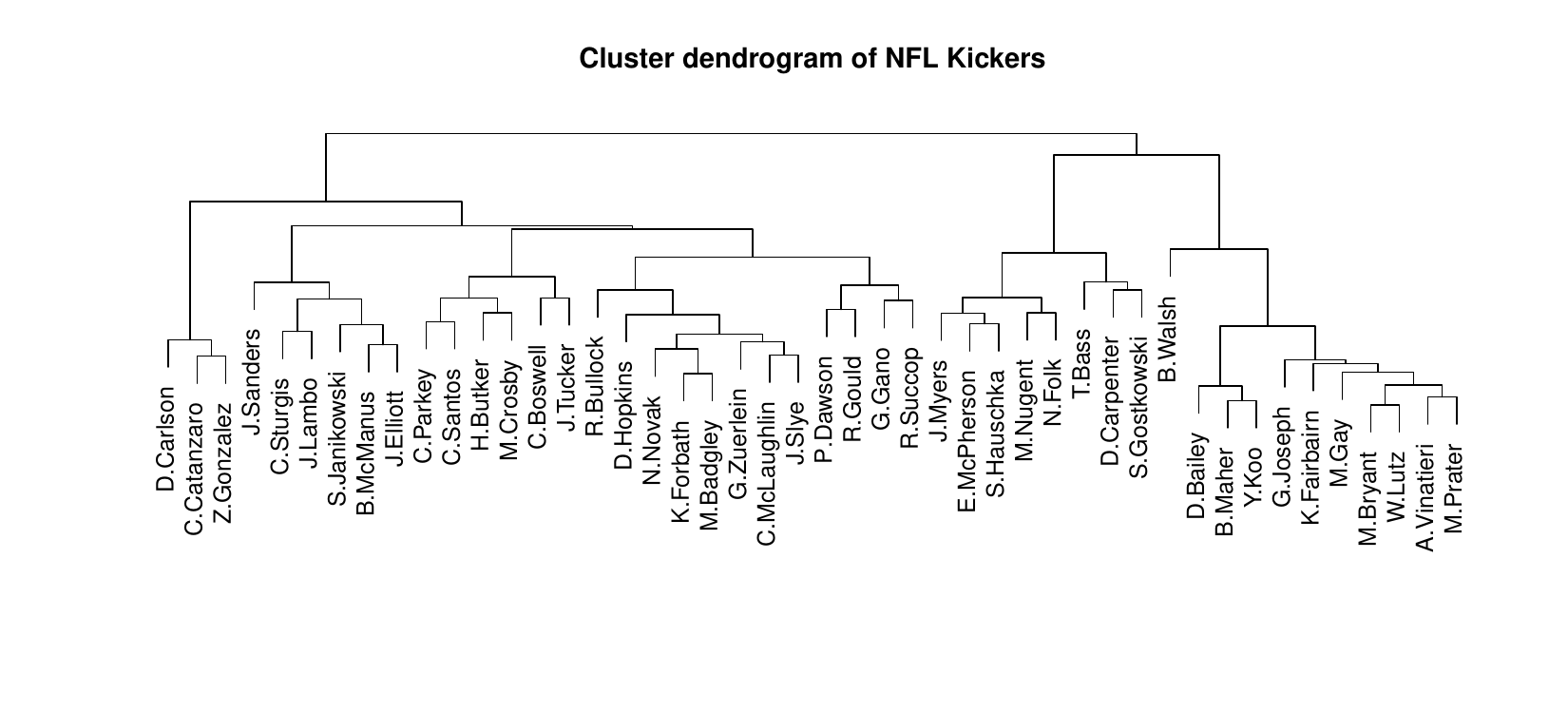}
    \caption{\small Cluster dendrogram of NFL kickers based on the Euclidean distance between their normalized propensity score estimates.}
    \label{fig:nfl-cluster-dendrogram}
\end{figure}

\subsection{MLB Batters}
In the second case study, we define and estimate evaluation metrics for MLB batters. The unit of analysis is an at-bat, and the outcome of interest is whether the batter obtained a legal hit of any kind. We use data from all 2025 regular season games, which we collected using the \texttt{baseballr} package \citep{petti2024}. The raw data comprise 163,664 at-bats taken by 673 unique players. We focus our analysis on players who reliably participate in every game (so-called everyday players), which we define as players who played in more than 150 games or had more than 600 plate appearances. After applying this constraint and filtering by complete cases, our analysis dataset contains 40,107 at-bats by 84 players with an overall batting average of .260.  

To characterize each at-bat, we include common measures of pitcher strength as covariates, including weighted on-base average against (wOBAA), strikeout rate, walk rate, hard-hit rate, barrel rate, swing rate, as well as pitcher handedness and the number of outs recorded prior to the at-bat. For each batter we estimated all three evaluation metrics. The required nuisance parameters $\mu_{\P}(a, X) = \E_{\P}(Y \mid A = a, X)$ and $\pi_{\P}(a \mid X) = \P(A = a \mid X)$ were estimated using a Super Learner ensemble with the \texttt{mean}, \texttt{glm}, and \texttt{lightgbm} learners and 5-fold cross-fitting. 

Assigning a causal interpretation to the estimated metrics rests on the validity of the identifying assumptions outlined in Section 2. When using observational data, there is the possibility of unmeasured confounding, i.e. the existence of relevant predictor variables which were not measured or not included in the analysis. This is certainly possible in this context, as we only incorporated measures of pitcher strength and game situation. For example, park factors like ballpark dimensions and weather conditions could influence the probability of obtaining a legal hit, but were not included in this analysis. The positivity assumption requires each batter to have a positive probability of taking every type of at-bat (in terms of the included predictors). This could be violated in practice, for example, if a particular batter only faced left-handed pitchers over the course of an entire season. We check empirically for positivity violations by analyzing the estimated propensity scores.   

After fitting the required nuisance estimators and checking the 3,368,988 estimated propensity scores, we do find some evidence for practical positivity violations. The mean propensity score is 0.0119, with 75 percent of the estimated scores between 0.0110 and 0.0127. The average probability for any batter in the dataset taking any particular attempt is just over 1 percent. The inverse weights computed from these values are likely increasing the width of the estimated 95\% confidence interval for each player.

In the context of the MLB case study, the direct standardization metric may be interpreted as the 2025 season batting average that would have been obtained by a particular player if all at-bats had been taken by that player. Because the reference distribution is the same across all batters, the direct standardization can be used to compare players. The estimated directly standardized batting averages for all 84 players, ranked in descending order, are shown in Figure \ref{fig:mlb-direct-standardization}. The top five players by directly standardized batting average are Aaron Judge, Freddie Freeman, Bo Bichette, Vladimir Guerrero Jr., and Xavier Edwards. The estimated confidence intervals for all 84 players contain their observed (non-standardized) batting averages. 

The leaderboard of standardized batting averages (Figure \ref{fig:mlb-direct-standardization}) ranks big names like Shohei Ohtani, Juan Soto and Kyle Schwarber outside the top 20. Recall the outcome of interest was obtaining a hit -- a single, double, triple and home run are all treated the same. Hitters like Schwarber are penalized, in a sense, by virtue of the outcome selected. Despite hitting 56 home runs, Schwarber also struck out 197 times, yielding an observed batting average of .240. The problems induced by collapsing performance, which is multi-faceted, to a single metric are well known, and there is no shortage of composite statistics in baseball (e.g. weighted on base average or on-base plus slugging) which attempt to account for different kinds of offensive capability. The novel methods described in this work can likewise be customized to rate offensive performance in alternative ways. 

Following the NFL kicker case study, we compute the contrasts $\widehat{\Delta}^\mathsf{indirect}_a$ and $\widehat{\Delta}^\mathsf{rand}_a$ for each batter. Indirect standardization measures the difference between a particular batter's empirical batting average and the estimated batting average that would have been observed had their at-bats been randomly alsigned to players with similar at-bat profiles (left-hand side of Figure \ref{fig:mlb-funnels}). We emphasize again that each estimate in this case is player specific and therefore cannot be directly compared across batters.

To make the notion of ``similar at-bat profiles'' more concrete, we applied a clustering procedure to the estimated propensity scores, following the NFL case study (Figure \ref{fig:mlb-cluster-dendrogram}). A noticeable feature of the resulting clustering is the isolation of Bryce Harper in the dendrogram. To investigate this, we compared the distribution of Harper's propensity scores with several other players (Appendix Figure~\ref{fig:mlb-harper}).  Harper's distribution is distinct from that of other batters in that it is noticeably more skewed. We further explored this by conducting a descriptive analysis of Bryce Harper's at-bats. Specifically, we fit a logistic classifier to predict whether a given at-bat was taken by Bryce Harper versus any other batter. For comparison, we fit an analogous model for Pete Alonso. For Harper, the coefficient for pitcher handedness had a larger magnitude than for Alonso, suggesting handedness was a significant predictor of a Harper at-bat, but not an Alonso at-bat. Examining the empirical proportion of at-bats against right-handed pitchers revealed that Harper had the lowest proportion of right-handed opponents in the analysis dataset. This could reflect teams strategically deploying left-handed pitchers against Harper, who bats left-handed, or it may simply be an artifact of the single-season sample. From the perspective of interpreting the indirect standardization parameter, these descriptive results indicate that Harper is compared with greater weight to other batters who also faced relative more left-handed pitchers.

Finally, the contrast between players' empirical batting averages and their replacement standardization metrics is shown in the right-hand side of Figure~\ref{fig:mlb-funnels}. This metric is interpreted as the difference between each players empirical batting average and the estimated batting average which would have been observed for a particular player, had their at-bats been re-assigned at uniformly at random to any player in the analysis dataset. In contrast to the NFL kicker analysis, several MLB batters have estimated indirect and replacement contrasts that are below zero with statistical significance. Unlike the NFL kickers, whose role is narrowly defined, most of the MLB players must also play defense (the exception being designated hitters). It often happens that players can find a spot in the everyday lineup even if they do not hit well on average, and rather fill a defensive need or can hit for moderate power (for example, Willy Adames and Taylor Ward). Injuries which are not season-ending but hamper performance (e.g. Anthony Volpe and Lars Nootbaar) may also be a contributing factor to lower estimated evaluation metrics. 

\begin{figure}
    \centering
    \includegraphics[width=\linewidth]{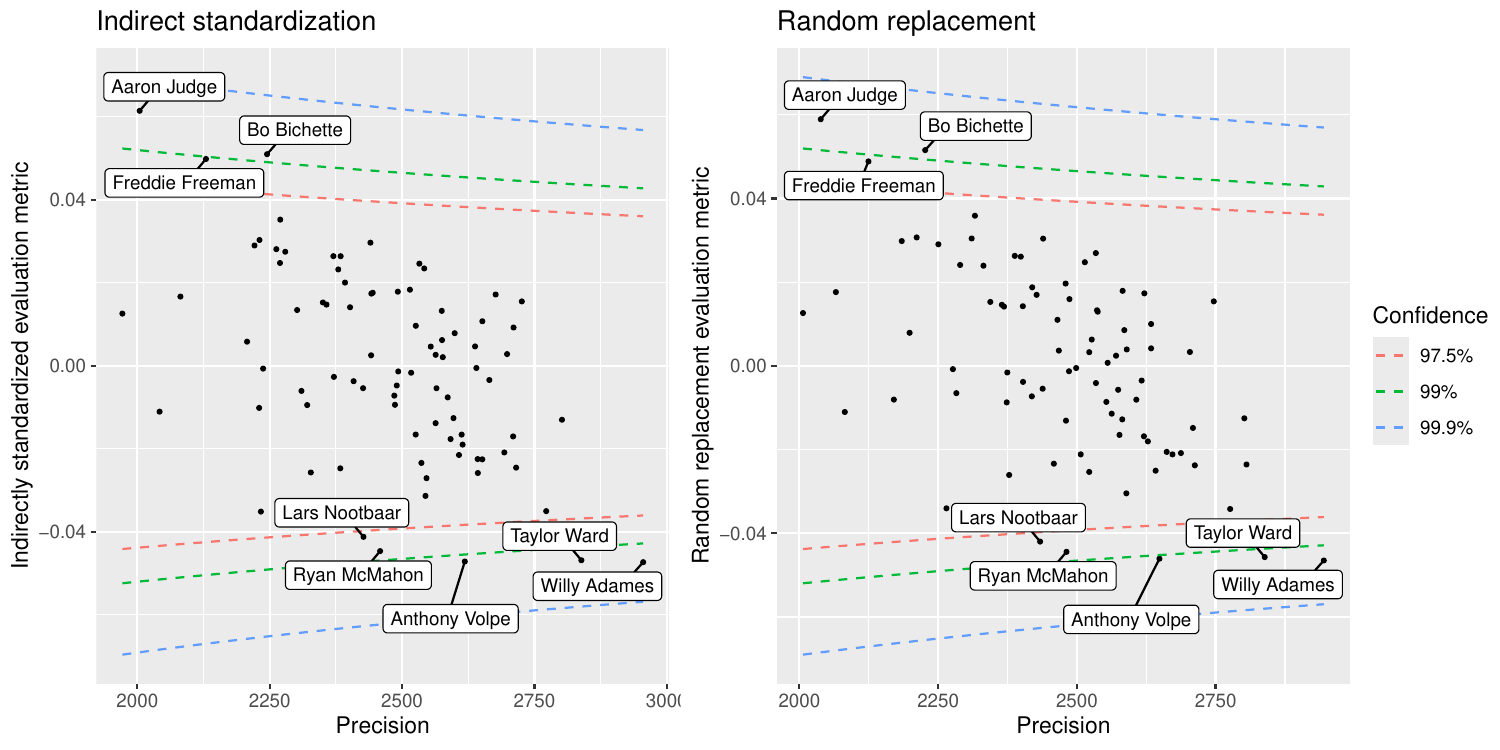} 
    \caption{\small Funnel plots visualizing the estimated indirect standardization (left) and random replacement (right) metrics for each MLB batter in the case study. The $x$-axis shows the statistical precision of the point estimate, and the $y$-axis is the point estimate. Dashed lines indicate control limits at the 97.5\%, 99\% and 99.9\% confidence levels.}
    \label{fig:mlb-funnels}
\end{figure}

\begin{figure}
    \centering
    \includegraphics[width=0.9\linewidth]{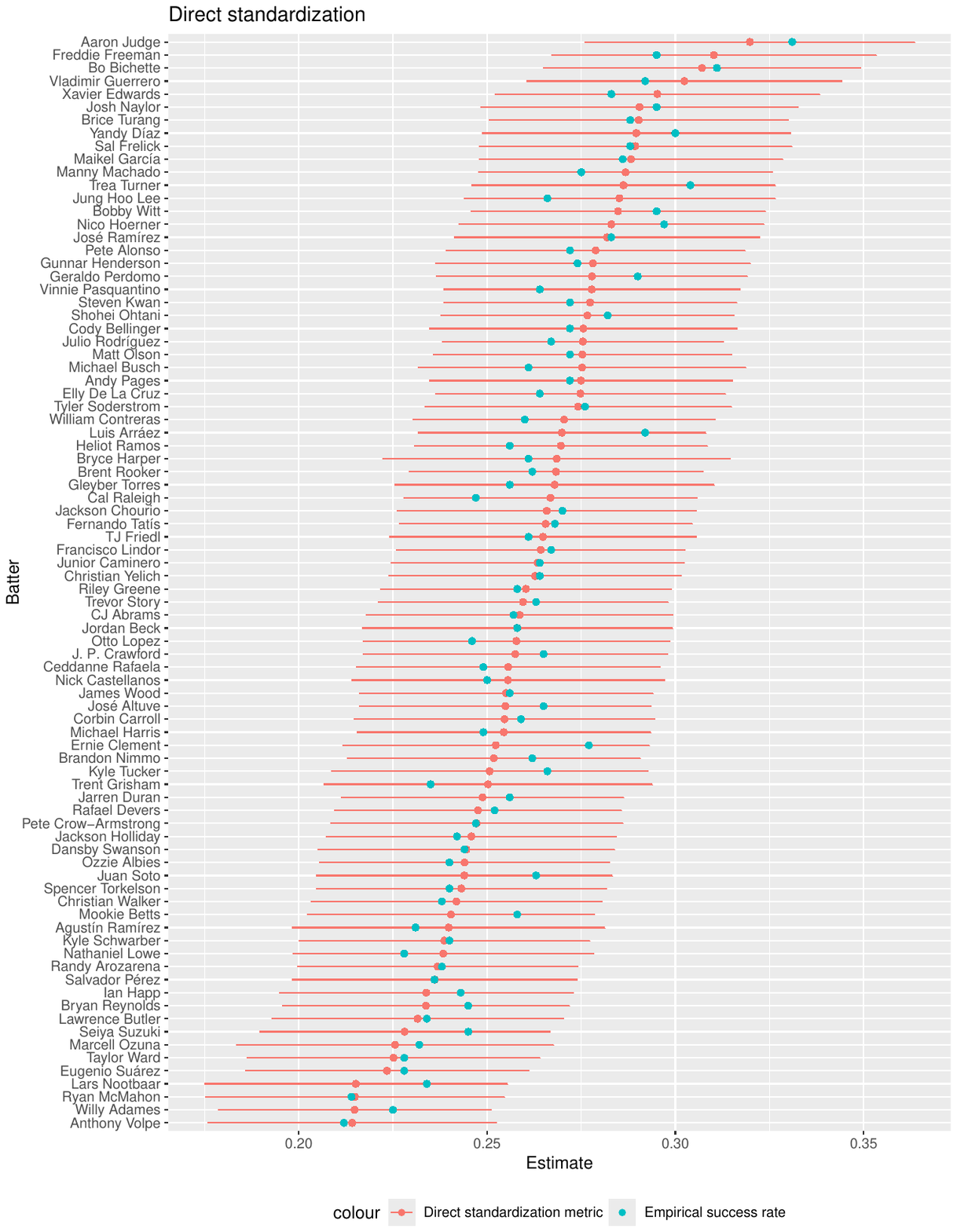}
    \caption{\small Directly standardized batting averages for each MLB batter in the analysis dataset.  The red points indicate the point estimate and the horizontal lines depict the 95\% confidence interval. The blue points show the empirical batting average for 2025. Batters are arranged in descending order of their point estimate.}
    \label{fig:mlb-direct-standardization}
\end{figure}

\begin{figure}
    \centering
    \includegraphics[width=0.95\linewidth]{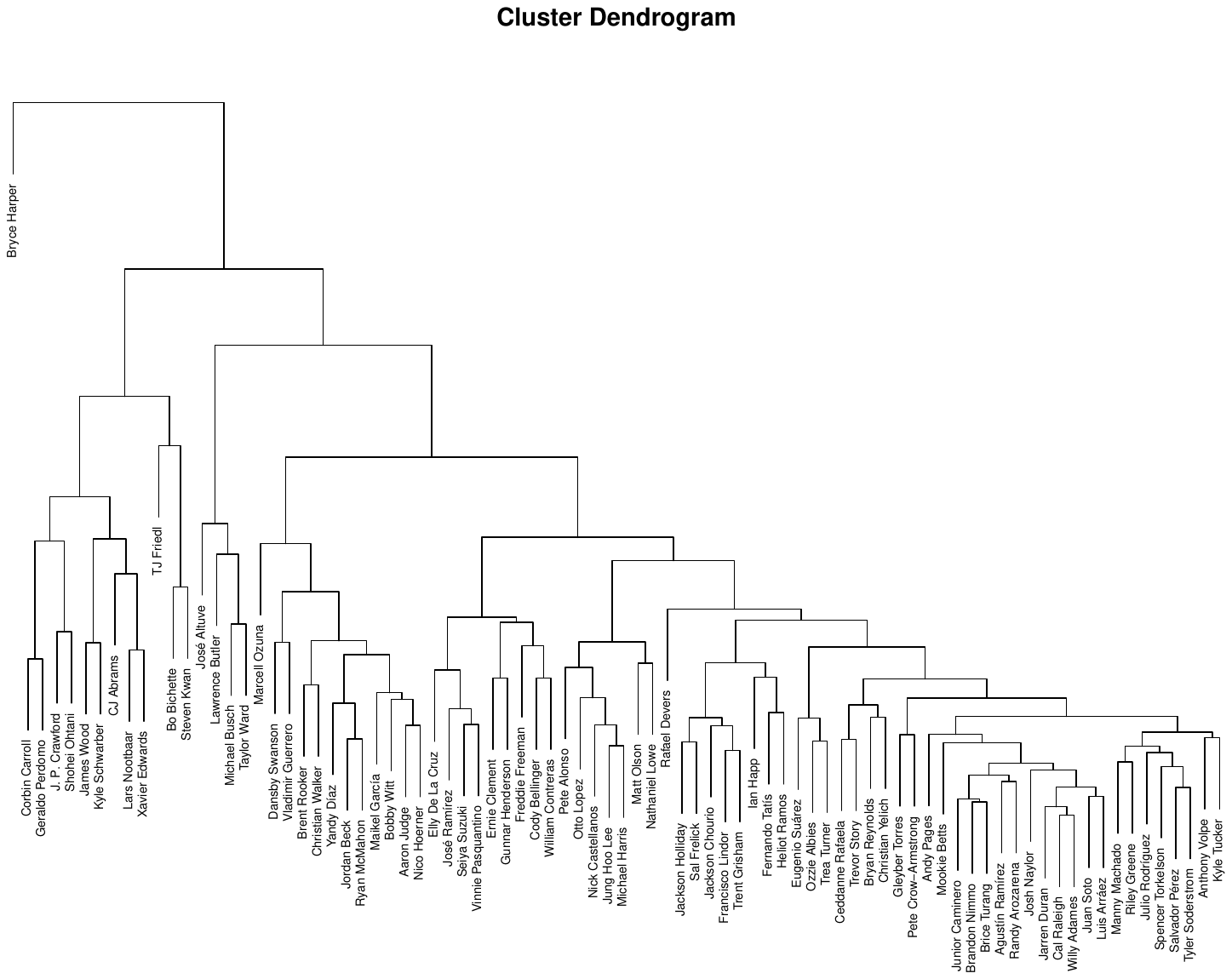}
    \caption{\small Cluster dendrogram of MLB batters based on the Euclidean distance between their normalized propensity score estimates.}
    \label{fig:mlb-cluster-dendrogram}
\end{figure}

\section{Discussion}
In this article, we posed player evaluation based on success over a series of attempts as a causal inference problem and proposed a general class of causal evaluation metrics. We investigated several estimands, beginning with direct and indirect standardization parameters that arise in healthcare provider profiling. We also introduced a ``performance above random replacement'' metric, which lies between direct and indirect standardization. Like indirect standardization, the random replacement metric standardizes to each player's attempt distribution, so players cannot be directly compared to one another. However, its interpretation is simpler than that of indirect parameters, because the stochastic intervention imagines reassigning attempts uniformly at random to any player in the league, rather than favoring players with similar attempt profiles. Additional estimands can be defined through alternative choices of stochastic intervention and conditioning sets, allowing metrics to be tailored to specific use cases and target audiences.

The case studies highlight important differences between healthcare provider profiling and sports player evaluation. In provider profiling, positivity violations are often severe, because providers can have highly specialized patient populations (as an extreme example, pediatric hospitals never treat the adult population). This lack of overlap makes estimation of direct standardization parameters untenable in many settings, which is why indirect standardization is widely used despite its more challenging interpretation. In our sports case studies, however, although players had different attempt profiles, we did not observe pronounced positivity violations, making direct standardization and random replacement parameters feasible to estimate in practice. Whether this holds in other sports settings is an interesting question for future research.

Our estimation approach focused on modern nonparametric inference methods. However, particularly in sports settings, the choice of estimator may depend on the goals and intended use of the metric. If minimizing bias is paramount, then using flexible ``black-box'' algorithms as inputs to TMLE may be warranted. For example, this might be the case for team managers seeking to make decisions based on the best available performance measures. On the other hand, if interpretability of the underlying prediction models is important, then a substitution estimator based on a parametric outcome model (such as logistic regression) may be preferable, for instance for fans or journalists who wish to to straightforwardly relate input data to the output metric. By separating the causal definition of our estimands from the estimation strategy, our framework allows for different choices of estimation strategy depending on the context. 

By framing player evaluation as a causal inference problem and relating it to healthcare provider profiling, this work opens the door to importing and adapting additional techniques from those literatures. For example, causal sensitivity analysis techniques \citep{robins2000sensitivity, diaz2013sensitivity, nabi2024sensitivity} could be applied to understand how unmeasured confounding or positivity violations affect player evaluation metrics. From the provider profiling side , statistical shrinkage techniques developed in that context \citep{mackenzie2015shrinkage, susmann2025penalized} could be combined with existing shrinkage-based approaches in sports analytics \citep{efron1977stein, osborne2017}, potentially yielding player ratings with better statistical properties while preserving a causal interpretation.


\bibliography{references}

@book{pearl2009, 
    place={Cambridge}, 
    edition={2}, 
    title={Causality}, 
    publisher={Cambridge University Press}, 
    author={Pearl, Judea}, 
    year={2009}
}

@book{vanderLaanRose11,
	Address = {New York},
	Author = {{van der Laan}, Mark J and Rose, Sherri},
	Publisher = {Springer},
	Title = {Targeted Learning: Causal Inference for Observational and Experimental Data},
	Year = 2011}

@article{mackenzie2015shrinkage,
title = {A Primer on Using Shrinkage to Compare In-Hospital Mortality Between Centers},
journal = {The Annals of Thoracic Surgery},
volume = {99},
number = {3},
pages = {757-761},
year = {2015},
issn = {0003-4975},
doi = {https://doi.org/10.1016/j.athoracsur.2014.11.039},
url = {https://www.sciencedirect.com/science/article/pii/S0003497514022036},
author = {Todd A. MacKenzie and Gary L. Grunkemeier and Gary K. Grunwald and A. James O’Malley and Chad Bohn and YingXing Wu and David J. Malenka}
}

@Article{daignault2019mediation,
author={Daignault, Katherine
and Lawson, Keith A.
and Finelli, Antonio
and Saarela, Olli},
title={Causal Mediation Analysis for Standardized Mortality Ratios},
journal={Epidemiology},
year={2019},
volume={30},
number={4},
issn={1044-3983},
url={https://journals.lww.com/epidem/fulltext/2019/07000/causal_mediation_analysis_for_standardized.10.aspx}
}

@article{shahian2020misuse,
author={Shahian, David M.
and Kozower, Benjamin D.
and Fernandez, Felix G.
and Badhwar, Vinay
and O'Brien, Sean M.},
title={The Use and Misuse of Indirectly Standardized, Risk-Adjusted Outcomes and Star Ratings},
journal={The Annals of Thoracic Surgery},
year={2020},
month={May},
day={01},
publisher={Elsevier},
volume={109},
number={5},
pages={1319-1322},
issn={0003-4975},
doi={10.1016/j.athoracsur.2019.09.010},
url={https://doi.org/10.1016/j.athoracsur.2019.09.010}
}

@article{normand1997profiling,
author = {Sharon-Lise T. Normand, Mark E. Glickman and Constantine A. Gatsonis},
title = {Statistical Methods for Profiling Providers of Medical Care: Issues and Applications},
journal = {Journal of the American Statistical Association},
volume = {92},
number = {439},
pages = {803--814},
year = {1997},
publisher = {Taylor \& Francis},
doi = {10.1080/01621459.1997.10474036},
URL = {https://doi.org/10.1080/01621459.1997.10474036},
eprint = {https://doi.org/10.1080/01621459.1997.10474036}
}

@inproceedings{clark2013going,
  title={Going for three: Predicting the likelihood of field goal success with logistic regression},
  author={Clark, Torin K. and Johnson, Aaron W. and Stimpson, Alexander J.},
  booktitle={The 7th Annual MIT Sloan Sports Analytics Conference},
  year={2013}
}

@article{bilder1998placekicks,
    author = {Christopher R. Bilder and Thomas M. Loughin},
    title = {“{It's} {Good!}” an Analysis of the Probability of Success for Placekicks},
    journal = {CHANCE},
    volume = {11},
    number = {2},
    pages = {20--30},
    year = {1998},
    publisher = {Taylor \& Francis},
    doi = {10.1080/09332480.1998.10542087},
    URL = {https://doi.org/10.1080/09332480.1998.10542087},
    eprint = {https://doi.org/10.1080/09332480.1998.10542087}
}

@article{berry1985rating,
author = {Donald A. Berry and Timothy D. Berry},
title = {The Probability of a Field Goal: Rating Kickers},
journal = {The American Statistician},
volume = {39},
number = {2},
pages = {152--155},
year = {1985},
publisher = {Taylor \& Francis},
doi = {10.1080/00031305.1985.10479418},
URL = {https://www.tandfonline.com/doi/abs/10.1080/00031305.1985.10479418},
eprint = {https://www.tandfonline.com/doi/pdf/10.1080/00031305.1985.10479418}
}

@article{berry2004,
author = {Scott M. Berry and Craig Wood},
title = {A Statistician Reads the Sports Pages},
journal = {CHANCE},
volume = {17},
number = {4},
pages = {47--51},
year = {2004},
publisher = {Taylor \& Francis},
doi = {10.1080/09332480.2004.10554926},
URL = {https://doi.org/10.1080/09332480.2004.10554926},
eprint = {https://doi.org/10.1080/09332480.2004.10554926}
}

@inbook{pasteur2016evaluation,
    author = {Pasteur, R. Drew and David, John A.},
    title = {Evaluation of Quarterbacks and Kickers},
    publisher = {Chapman and Hall/CRC},
    year = {2016},
    chapter = {7},
    booktitle = {Handbook of Statistical Methods and Analyses in Sports},
    address={New York}
}

@Manual{ho2025nflreadr,
    title = {nflreadr: Download 'nflverse' Data},
    author = {Tan Ho and Sebastian Carl},
    year = {2025},
    note = {R package version 1.5.0},
    url = {https://CRAN.R-project.org/package=nflreadr},
    doi = {10.32614/CRAN.package.nflreadr},
  }

@InProceedings{bacco2024plusminus,
  title = 	 {A causality-inspired  plus-minus model for player evaluation in team sports},
  author =       {Bacco, Caterina De and Wang, Yixin and Blei, David},
  booktitle = 	 {Proceedings of the Third Conference on Causal Learning and Reasoning},
  pages = 	 {769--792},
  year = 	 {2024},
  editor = 	 {Locatello, Francesco and Didelez, Vanessa},
  volume = 	 {236},
  series = 	 {Proceedings of Machine Learning Research},
  month = 	 {01--03 Apr},
  publisher =    {PMLR},
  url = 	 {https://proceedings.mlr.press/v236/bacco24a.html}
}

@article{susmann2025providerevaluation,
    author = {Susmann, Herbert P and Li, Yiting and McAdams-DeMarco, Mara A and Díaz, Iván and Wu, Wenbo},
    title = {Doubly robust nonparametric efficient estimation for healthcare provider evaluation},
    journal = {Journal of the Royal Statistical Society Series A: Statistics in Society},
    pages = {qnaf145},
    year = {2025},
    month = {10},
    issn = {0964-1998},
    doi = {10.1093/jrsssa/qnaf145},
    url = {https://doi.org/10.1093/jrsssa/qnaf145},
    eprint = {https://academic.oup.com/jrsssa/advance-article-pdf/doi/10.1093/jrsssa/qnaf145/64567718/qnaf145.pdf},
}

@article{daignault2017doublerobust,
    url = {https://doi.org/10.1515/em-2016-0016},
    title = {Doubly Robust Estimator for Indirectly Standardized Mortality Ratios},
    author = {Katherine Daignault and Olli Saarela},
    pages = {20160016},
    volume = {6},
    number = {1},
    journal = {Epidemiologic Methods},
    doi = {doi:10.1515/em-2016-0016},
    year = {2017},
    lastchecked = {2024-03-14}
}

@book{vanderVaart98,
	Author = {A. W. van der Vaart},
	Pages = {62},
	Publisher = {Cambridge University Press},
	Title = {Asymptotic Statistics},
	Year = {1998}}

@misc{kennedy2023semiparametric,
      title={Semiparametric doubly robust targeted double machine learning: a review}, 
      author={Edward H. Kennedy},
      year={2023},
      eprint={2203.06469},
      archivePrefix={arXiv},
      primaryClass={stat.ME}
}

@article{fisher2021,
author = {Aaron Fisher and Edward H. Kennedy},
title = {Visually Communicating and Teaching Intuition for Influence Functions},
journal = {The American Statistician},
volume = {75},
number = {2},
pages = {162--172},
year = {2021},
publisher = {Taylor \& Francis},
doi = {10.1080/00031305.2020.1717620},
URL = {https://doi.org/10.1080/00031305.2020.1717620},
eprint = {https://doi.org/10.1080/00031305.2020.1717620}
}

@book{bickel1993efficient,
  title={Efficient and adaptive estimation for semiparametric models},
  author={Bickel, Peter J and Klaassen, Chris AJ and Ritov, Ya’acov and Klaassen, J and Wellner, Jon A},
  volume={4},
  year={1993},
  publisher={Springer}
}

@book{vanderVaart1996,
author="van der Vaart, Aad W.
and Wellner, Jon A.",
title="Weak Convergence and Empirical Processes: With Applications to Statistics",
year="1996",
publisher="Springer New York",
address="New York, NY"
}

@Inbook{zheng2011cvtmle,
author="Zheng, Wenjing
and van der Laan, Mark J.",
title="Cross-Validated Targeted Minimum-Loss-Based Estimation",
bookTitle="Targeted Learning: Causal Inference for Observational and Experimental Data",
year="2011",
publisher="Springer New York",
address="New York, NY",
pages="459--474",
isbn="978-1-4419-9782-1",
doi="10.1007/978-1-4419-9782-1{\_}27",
url="https://doi.org/10.1007/978-1-4419-9782-1\_27"
}

@article{vanderlaan2007superlearner,
url = {https://doi.org/10.2202/1544-6115.1309},
title = {Super Learner},
author = {Mark J. van der Laan and Eric C Polley and Alan E. Hubbard},
volume = {6},
number = {1},
journal = {Statistical Applications in Genetics and Molecular Biology},
doi = {doi:10.2202/1544-6115.1309},
year = {2007},
lastchecked = {2024-06-06}
}

@article{vanderLaan2006tmle,
url = {https://doi.org/10.2202/1557-4679.1043},
title = {Targeted Maximum Likelihood Learning},
author = {Mark J. van der Laan and Daniel Rubin},
volume = {2},
number = {1},
journal = {The International Journal of Biostatistics},
doi = {doi:10.2202/1557-4679.1043},
year = {2006},
lastchecked = {2024-06-06}
}

@article{chernozhukov2018dml,
    author = {Chernozhukov, Victor and Chetverikov, Denis and Demirer, Mert and Duflo, Esther and Hansen, Christian and Newey, Whitney and Robins, James},
    title = "{Double/debiased machine learning for treatment and structural parameters}",
    journal = {The Econometrics Journal},
    volume = {21},
    number = {1},
    pages = {C1-C68},
    year = {2018},
    month = {01},
    issn = {1368-4221},
    doi = {10.1111/ectj.12097},
    url = {https://doi.org/10.1111/ectj.12097},
    eprint = {https://academic.oup.com/ectj/article-pdf/21/1/C1/27684918/ectj00c1.pdf},
}

@book{pfanzagl1982,
author="Pfanzagl, J.",
title="Contributions to a General Asymptotic Statistical Theory",
year="1982",
publisher="Springer New York",
address="New York, NY"
}

@article{schick1986semiparametric,
    author = {Anton Schick},
    title = {On Asymptotically Efficient Estimation in Semiparametric Models},
    volume = {14},
    journal = {The Annals of Statistics},
    number = {3},
    publisher = {Institute of Mathematical Statistics},
    pages = {1139 -- 1151},
    year = {1986},
    doi = {10.1214/aos/1176350055},
    URL = {https://doi.org/10.1214/aos/1176350055}
}

@inbook{robins2008higherorder,
   title={Higher order influence functions and minimax estimation of nonlinear functionals},
   ISBN={0940600749},
   url={http://dx.doi.org/10.1214/193940307000000527},
   DOI={10.1214/193940307000000527},
   booktitle={Probability and Statistics: Essays in Honor of David A. Freedman},
   publisher={Institute of Mathematical Statistics},
   author={Robins, James and Li, Lingling and Tchetgen, Eric and van der Vaart, Aad},
   year={2008},
   pages={335–421}
}

@Inbook{Hubbard2011,
author="Hubbard, Alan E.
and Jewell, Nicholas P.
and van der Laan, Mark J.",
title="Direct Effects and Effect Among the Treated",
bookTitle="Targeted Learning: Causal Inference for Observational and Experimental Data",
year="2011",
publisher="Springer New York",
address="New York, NY",
pages="133--143",
isbn="978-1-4419-9782-1",
doi="10.1007/978-1-4419-9782-1_8",
url="https://doi.org/10.1007/978-1-4419-9782-1_8"
}

@article{pasteur2014expectation,
    url = {https://doi.org/10.1515/jqas-2013-0039},
    title = {An expectation-based metric for {NFL} field goal kickers},
    author = {R. Drew Pasteur and Kyle Cunningham-Rhoads},
    pages = {49--66},
    volume = {10},
    number = {1},
    journal = {Journal of Quantitative Analysis in Sports},
    doi = {doi:10.1515/jqas-2013-0039},
    year = {2014},
    lastchecked = {2026-02-17}
}

@article{baumer2015,
    url = {https://doi.org/10.1515/jqas-2014-0098},
    title = {{openWAR}: An open source system for evaluating
overall player performance in major league
baseball},
    author = {Baumer, B. and Jensen, S.T. and Matthews, G.J.},
    pages = {69--84},
    volume = {11},
    number = {2},
    journal = {Journal of Quantitative Analysis in Sports},
    doi = {},
    year = {2015},
    lastchecked = {2026-02-18}
}

@book{kosorok2008,
  title     = "Introduction to empirical processes and semiparametric inference",
  author    = "Kosorok, Michael R",
  publisher = "Springer",
  series    = "Springer Series in Statistics",
  edition   =  2008,
  month     =  dec,
  year      =  2008,
  address   = "New York, NY",
  language  = "en"
}

@article{osborne2017,
    url = {https://doi.org/10.3233/JSA-16140},
    title = {Shrinkage estimation of NFL field goal success probabilities},
    author = {Osborne, J.A. and Levine, R.},
    pages = {129--146},
    volume = {3},
    number = {2},
    journal = {Journal of Sports Analytics},
    doi = {},
    year = {2017},
    lastchecked = {2026-02-18}
}

@book{iezzoni2003,
    author={Iezzoni, Lisa I.},
    title={Risk adjustment for measuring health care outcomes},
    year={2003},
    edition={3rd ed.},
    publisher={Health Administration Press},
    address={Chicago, IL},
    note={Includes bibliographical references (p. 411-478) and index.},
    isbn={9781567932072}
}

@article{petersen2012positivity,
author = {Maya L Petersen and Kristin E Porter and Susan Gruber and Yue Wang and Mark J van der Laan},
title ={Diagnosing and responding to violations in the positivity assumption},
    journal = {Statistical Methods in Medical Research},
    volume = {21},
    number = {1},
    pages = {31-54},
    year = {2012},
    doi = {10.1177/0962280210386207},
        note ={PMID: 21030422},
    URL = {https://doi.org/10.1177/0962280210386207
    },
    eprint = {https://doi.org/10.1177/0962280210386207}
}

@article{terner2021basketball,
   author = "Terner, Zachary and Franks, Alexander",
   title = "Modeling Player and Team Performance in Basketball", 
   journal= "Annual Review of Statistics and Its Application",
   year = "2021",
   volume = "8",
   number = "Volume 8, 2021",
   pages = "1-23",
   doi = "https://doi.org/10.1146/annurev-statistics-040720-015536",
   url = "https://www.annualreviews.org/content/journals/10.1146/annurev-statistics-040720-015536",
   publisher = "Annual Reviews",
   issn = "2326-831X",
   type = "Journal Article",
}

@article{yam2019fourthdown,
author = {Derrick R. Yam and Michael J. Lopez},
title ={What was lost? A causal estimate of fourth down behavior in the {National} {Football} {League}},
journal = {Journal of Sports Analytics},
volume = {5},
number = {3},
pages = {153-167},
year = {2019},
doi = {10.3233/JSA-190294},
URL = {https://doi.org/10.3233/JSA-190294},
eprint = {https://doi.org/10.3233/JSA-190294}
}

@article{efron1977stein,
    address = {New York},
    issn = {0036-8733},
    journal = {Scientific American},
    lccn = {2006255042},
    number = {5},
    publisher = {Rufus Porter},
    title = {Stein's paradox in statistics},
    author={Bradley Efron and Carol Morris},
    volume = {236},
    year = {1977},
}

@article{koseler2017machinelearningbaseball,
    author = {Kaan Koseler and Matthew Stephan},
    title = {Machine Learning Applications in Baseball: A Systematic Literature Review},
    journal = {Applied Artificial Intelligence},
    volume = {31},
    number = {9-10},
    pages = {745--763},
    year = {2017},
    publisher = {Taylor \& Francis},
    doi = {10.1080/08839514.2018.1442991},
    URL = {https://doi.org/10.1080/08839514.2018.1442991},
    eprint = {https://doi.org/10.1080/08839514.2018.1442991}
}

@article{sanchez2024icing,
    author = {Adriana {Gonzalez Sanchez} and Sierra Martinez and Ron Yurko and Ryan Elmore and Brian Macdonald},
    title = {Beyond the Box Score: Does Icing the Field Goal Kicker Work in the {NFL}?},
    journal = {CHANCE},
    volume = {37},
    number = {3},
    pages = {41--48},
    year = {2024},
    publisher = {Taylor \& Francis},
    doi = {10.1080/09332480.2024.2415841},
    URL = {https://doi.org/10.1080/09332480.2024.2415841},
    eprint = {https://doi.org/10.1080/09332480.2024.2415841}
}

@BOOK{fujii2025machinelearningsports,
  title     = "Machine learning in sports",
  author    = "Fujii, Keisuke",
  publisher = "Springer Nature",
  month     =  apr,
  year      =  2025,
  address   = "Cham, Switzerland",
  copyright = "https://creativecommons.org/licenses/by/4.0",
  language  = "en"
}

@BOOK{brefeld2022machinelearningsports,
  title     = "Machine learning and data mining for sports analytics",
  editor    = "Brefeld, Ulf and Davis, Jesse and Van Haaren, Jan and
               Zimmermann, Albrecht",
  publisher = "Springer International Publishing",
  series    = "Communications in Computer and Information Science",
  edition   =  1,
  month     =  may,
  year      =  2022,
  address   = "Cham, Switzerland",
  language  = "en"
}

@article{weimer2023momentum,
    author = {Louis Weimer and Zachary C. Steinert-Threlkeld and Kevin Coltin},
    title ={A causal approach for detecting team-level momentum in {NBA} games},
    journal = {Journal of Sports Analytics},
    volume = {9},
    number = {2},
    pages = {117-132},
    year = {2023},
    doi = {10.3233/JSA-220592},
    URL = {https://doi.org/10.3233/JSA-220592},
    eprint = {https://doi.org/10.3233/JSA-220592}
}

@inproceedings{assis2020timeout,
author = {Assis, Niander and Assun\c{c}\~{a}o, Renato and Vaz-de-Melo, Pedro O. S.},
title = {Stop the Clock: Are Timeout Effects Real?},
year = {2020},
isbn = {978-3-030-67669-8},
publisher = {Springer-Verlag},
address = {Berlin, Heidelberg},
url = {https://doi.org/10.1007/978-3-030-67670-4_31},
doi = {10.1007/978-3-030-67670-4_31},
booktitle = {Machine Learning and Knowledge Discovery in Databases. Applied Data Science and Demo Track: European Conference, ECML PKDD 2020, Ghent, Belgium, September 14–18, 2020, Proceedings, Part V},
pages = {507–523},
numpages = {17},
location = {Ghent, Belgium}
}

@article{gibbs2022timeout,
    author = {Connor P. Gibbs and Ryan Elmore and Bailey K. Fosdick},
    title = {{The causal effect of a timeout at stopping an opposing run in the NBA}},
    volume = {16},
    journal = {The Annals of Applied Statistics},
    number = {3},
    publisher = {Institute of Mathematical Statistics},
    pages = {1359 -- 1379},
    year = {2022},
    doi = {10.1214/21-AOAS1545},
    URL = {https://doi.org/10.1214/21-AOAS1545}
}

@ARTICLE{kubatko2007basketball,
  title     = "A starting point for analyzing basketball statistics",
  author    = "Kubatko, Justin and Oliver, Dean and Pelton, Kevin and
               Rosenbaum, Dan T",
  journal   = "J. Quant. Anal. Sports",
  publisher = "Walter de Gruyter GmbH",
  volume    =  3,
  number    =  3,
  month     =  jan,
  year      =  2007
}

@article{macdonald2011plusminus,
    url = {https://doi.org/10.2202/1559-0410.1284},
    title = {A Regression-Based Adjusted Plus-Minus Statistic for {NHL} Players},
    author = {Brian Macdonald},
    volume = {7},
    number = {3},
    journal = {Journal of Quantitative Analysis in Sports},
    doi = {doi:10.2202/1559-0410.1284},
    year = {2011},
    lastchecked = {2026-02-19}
}

@misc{mcclean2025comparingcausalparameterstreatments,
      title={Comparing causal parameters with many treatments and positivity violations}, 
      author={Alec McClean and Yiting Li and Sunjae Bae and Mara A. McAdams-DeMarco and Iván Díaz and Wenbo Wu},
      year={2025},
      eprint={2410.13522},
      archivePrefix={arXiv},
      primaryClass={stat.ME},
      url={https://arxiv.org/abs/2410.13522}, 
}

@Manual{petti2024,
  title = {baseballr: Acquiring and Analyzing Baseball Data},
  author = {Bill Petti and Saiem Gilani},
  year = {2024},
  note = {R package version 1.6.0, 
https://github.com/BillPetti/baseballr},
  url = {https://billpetti.github.io/baseballr/},
}

@Article{arelbundock2024marginaleffects,
    title = {How to Interpret Statistical Models Using {marginaleffects} for {R} and {Python}},
    author = {Vincent Arel-Bundock and Noah Greifer and Andrew Heiss},
    journal = {Journal of Statistical Software},
    year = {2024},
    volume = {111},
    number = {9},
    pages = {1--32},
    doi = {10.18637/jss.v111.i09},
  }

@article{christiansen1997provider,
title = {Improving the Statistical Approach to Health Care Provider Profiling},
journal = {Annals of Internal Medicine},
author={Cindy L. Christiansen and Carl N. Morris},
volume = {127},
number = {8\_Part\_2},
pages = {764-768},
year = {1997},
doi = {10.7326/0003-4819-127-8\_Part\_2-199710151-00065},
    note ={PMID: 9382395},
URL = {https://doi.org/10.7326/0003-4819-127-8_Part_2-199710151-00065},
eprint = {https://doi.org/10.7326/0003-4819-127-8_Part_2-199710151-00065}
}

@article{delong1997profiling,
author = {DeLong, Elizabeth R. and Peterson, Eric D. and DeLong, David M. and Muhlbaier, Lawrence H. and Hackett, Suzanne and Mark, Daniel B.},
title = {Comparing risk-adjustment methods for provider profiling},
journal = {Statistics in Medicine},
volume = {16},
number = {23},
pages = {2645-2664},
doi = {https://doi.org/10.1002/(SICI)1097-0258(19971215)16:23<2645::AID-SIM696>3.0.CO;2-D},
url = {https://onlinelibrary.wiley.com/doi/abs/10.1002/%28SICI%291097-0258%2819971215%2916%3A23%3C2645%3A%3AAID-SIM696%3E3.0.CO%3B2-D},
eprint = {https://onlinelibrary.wiley.com/doi/pdf/10.1002/%28SICI%291097-0258%2819971215%2916%3A23%3C2645%3A%3AAID-SIM696%3E3.0.CO%3B2-D},
year = {1997}
}

@ARTICLE{diaz2013sensitivity,
  title    = "Sensitivity analysis for causal inference under unmeasured
              confounding and measurement error problems",
  author   = "D{\'\i}az, Iv{\'a}n and van der Laan, Mark J",
  journal  = "Int J Biostat",
  volume   =  9,
  number   =  2,
  pages    = "149--160",
  month    =  nov,
  year     =  2013,
  address  = "Germany",
  language = "en"
}

@article{griffen2012funnel,
author = {Griffen, David and Callahan, Charles D. and Markwell, Stephen and de la Cruz, Jonathan and Milbrandt, Joseph C. and Harvey, Timothy},
title = {Application of Statistical Process Control to Physician-specific Emergency Department Patient Satisfaction Scores: A Novel Use of the Funnel Plot},
journal = {Academic Emergency Medicine},
volume = {19},
number = {3},
pages = {348-355},
doi = {https://doi.org/10.1111/j.1553-2712.2012.01304.x},
url = {https://onlinelibrary.wiley.com/doi/abs/10.1111/j.1553-2712.2012.01304.x},
eprint = {https://onlinelibrary.wiley.com/doi/pdf/10.1111/j.1553-2712.2012.01304.x},
year = {2012}
}

@article{spiegelhalter2005funnel,
author = {Spiegelhalter, David J.},
title = {Funnel plots for comparing institutional performance},
journal = {Statistics in Medicine},
volume = {24},
number = {8},
pages = {1185-1202},
doi = {https://doi.org/10.1002/sim.1970},
url = {https://onlinelibrary.wiley.com/doi/abs/10.1002/sim.1970},
eprint = {https://onlinelibrary.wiley.com/doi/pdf/10.1002/sim.1970},
year = {2005}
}

@Article{breiman1996stacking,
author={Breiman, Leo},
title={Stacked regressions},
journal={Machine Learning},
year={1996},
month={Jul},
day={01},
volume={24},
number={1},
pages={49-64},
issn={1573-0565},
doi={10.1007/BF00117832},
url={https://doi.org/10.1007/BF00117832}
}

@article{rotnitzky2020mixedbias,
    author = {Rotnitzky, A and Smucler, E and Robins, J M},
    title = {Characterization of parameters with a mixed bias property},
    journal = {Biometrika},
    volume = {108},
    number = {1},
    pages = {231-238},
    year = {2020},
    month = {08},
    issn = {0006-3444},
    doi = {10.1093/biomet/asaa054},
    url = {https://doi.org/10.1093/biomet/asaa054},
    eprint = {https://academic.oup.com/biomet/article-pdf/108/1/231/36441108/asaa054.pdf},
}

@article{nabi2024sensitivity,
    author = {Nabi, Razieh and Bonvini, Matteo and Kennedy, Edward H and Huang, Ming-Yueh and Smid, Marcela and Scharfstein, Daniel O},
    title = {Semiparametric sensitivity analysis: unmeasured confounding in observational studies},
    journal = {Biometrics},
    volume = {80},
    number = {4},
    pages = {ujae106},
    year = {2024},
    month = {10},
    issn = {0006-341X},
    doi = {10.1093/biomtc/ujae106},
    url = {https://doi.org/10.1093/biomtc/ujae106},
    eprint = {https://academic.oup.com/biometrics/article-pdf/80/4/ujae106/61199725/ujae106.pdf},
}

@InProceedings{robins2000sensitivity,
    author="Robins, James M.
    and Rotnitzky, Andrea
    and Scharfstein, Daniel O.",
    editor="Halloran, M. Elizabeth
    and Berry, Donald",
    title="Sensitivity Analysis for Selection bias and unmeasured Confounding in missing Data and Causal inference models",
    booktitle="Statistical Models in Epidemiology, the Environment, and Clinical Trials",
    year="2000",
    publisher="Springer New York",
    address="New York, NY",
    pages="1--94",
    isbn="978-1-4612-1284-3"
}

@misc{susmann2025penalized,
      title={Asymptotically Efficient Data-adaptive Penalized Shrinkage Estimation with Application to Causal Inference}, 
      author={Herbert P. Susmann and Yiting Li and Mara A. McAdams-DeMarco and Wenbo Wu and Iván Díaz},
      year={2025},
      eprint={2505.08065},
      archivePrefix={arXiv},
      primaryClass={stat.ME},
      url={https://arxiv.org/abs/2505.08065}, 
}

\begin{appendix}

\section{Identification proof}
The identification result can be derived by the following standard argument:
\begin{align}
    &\E\left[Y(A^*) \mid X \in \mathcal{X}', A \in \mathcal{A}'\right] \\
    &= \E\left[\sum_{a' \in \mathcal{A}} Y\left(a'\right) \1[A^* = a'] \mid X \in \mathcal{X}', A \in \mathcal{A}'\right] \\
    &= \E\left[\sum_{a' \in \mathcal{A}} \E\left[Y\left(a'\right) \1[A^* = a'] \mid X\right] \mid X \in \mathcal{X}', A \in \mathcal{A}'\right] \text{\quad (law of iterated expectations)} \\
    &= \E\left[\sum_{a' \in \mathcal{A}} \E\left[Y\left(a'\right) \mid X\right] P(A^* = a' \mid X) \mid X \in \mathcal{X}', A \in \mathcal{A}'\right] \text{\quad (Assumption~\ref{assumption:intervention})} \\
    &= \E\left[\sum_{a' \in \mathcal{A}} \E\left[Y\left(a'\right) \mid A = a', X\right] P(A^* = a' \mid X) \mid X \in \mathcal{X}', A \in \mathcal{A}'\right] \text{\quad (Assumption~\ref{assumption:no-unmeasured-confounding})} \\
    &= \E\left[\sum_{a' \in \mathcal{A}} \E\left[Y \mid A = a', X\right] P(A^* = a' \mid X) \mid X \in \mathcal{X}', A \in \mathcal{A}'\right] \text{ \quad (Definition of $Y$  and Assumption~\ref{assumption:positivity}.)}
\end{align}

\section{Performance above random replacement}
\subsection{Semiparametric analysis}
\label{appendix:semi-parametrics}
Fix $a, a' \in \mathcal{A}$. For any $\P \in \mathcal{M}$, let $\theta_P(a, a') = \E[\mu(a', X) \mid A = a, X \in \mathcal{X}']$.
\begin{theorem}[von Mises expansion]
    \label{thm:vonmises}
    For any $\P \in \mathcal{M}$. For any $\P, \F \in \mathcal{M}$ and any $a, a' \in \mathcal{A}$, the parameter $\theta(a, a')$ characterized by $P \mapsto \E[\mu(a', X) \mid A = a, X \in \mathcal{X}']$ satisfies the expansion
    \begin{align}
        \theta_{\F}(a, a') - \theta_{\P}(a, a') = \E_{\P}\left\{ \D_{\F}(Z) \right\} + \R(\P, \F),
    \end{align}
    where $\D(\F) : \mathcal{Z} \to \mathbb{R}$ is characterized by
    \begin{align}
        \D(\F)(Z) = \frac{\1(X \in \mathcal{X}')}{\F(X \in \mathcal{X}')} \left[ \frac{\1(A = a')}{\F(A = a)} \frac{\pi_{\F}(a \mid X)}{\pi_{\F}(a' \mid X)} \left\{ Y - \mu_{\F}(a', X) \right\} + \frac{\1(A = a)}{\F(A = a)} \left\{ \mu_{\F}(a', X) - \theta_{\F}(a, a') \right\} \right]
    \end{align}
    and $\R : \mathcal{M} \times \mathcal{M} \to \mathbb{R}$ is characterized by
    \begin{align}
        \R(\P, \F) = \frac{\F(X \in \mathcal{X}')}{\P(X \in \mathcal{X}')} \E_{\F}\left[ \frac{\pi_{\P}(a)}{\P(a)} \left\{ \frac{\pi_{\F}(a')}{\pi_{\P}(a')} - \frac{\pi_{\F}(a)}{\pi_{\P}(a)} \right\} \left\{ \mu_{\F} - \mu_{\P} \right\} + \left\{ \frac{\F(X \in \mathcal{X}')}{\P(X \in \mathcal{X}')} \frac{\F(a)}{\P(a)} - 1\right\} \left\{\theta_{\P} - \theta_{\F}\right\} \right].
    \end{align}
\end{theorem}
\begin{proof}
    For notational succinctness, for any $\P \in \mathcal{M}$ we write $\theta_{\P} \equiv \theta_{\P}(a, a')$, $\pi_{\P}(a) \equiv \pi_{\P}(a \mid X)$, $\mu_{\P} \equiv \mu_{\P}(a', X)$, $w_{\P} \equiv w_{\P}(X)$, and $\P(a) \equiv \P(A = a)$ The second-order term $\R(\P, \F)$ is given by
    \begin{align}
        \R(\P, \F) = \theta_{\F} - \theta_{\P} + \textcolor{purple}{\E_{\P}\left\{ \D(\F)(Z) \right\}}.
    \end{align}
    The key step is to write
    \begin{align}
        \label{eq:second-order-remainder-long}
        \R(\P, \F) = \theta_{\F} - \theta_{\P} + \textcolor{purple}{\E_{\P}\left\{ \D(\F)(Z) \right\}}  - \textcolor{teal}{\E_{\P}\left\{ \D_{\P}(\F)(Z) \right\}} + \textcolor{teal}{\E_{\P}\left\{ \D_{\P}(\F)(Z) \right\}},
    \end{align}
   where $\D_{\P}(\F)$ is the EIF at $\F$ with the $\theta_{\F}$ replaced with $\theta_{\P}$:
    \begin{align}
        \D_{\P}(\F)(Z) &= \frac{\1(X \in \mathcal{X}')}{\F(X \in \mathcal{X}')} \left[ \frac{\1(A = a')}{\F(a)} \frac{\pi_{\F}(a)}{\pi_{\F}(a')} \left\{ Y - \mu_{\F} \right\} + \frac{\1(A = a)}{\F(a)} \left\{ \mu_{\F} - \theta_{\P} \right\} \right].
    \end{align}
    Then its expectation with respect to $\P$ simplifies conveniently:
    \begin{align}
        \textcolor{teal}{\E_{\P}\left\{ \D_{\P}(\F)(Z) \right\}} &=  \frac{\P(X \in \mathcal{X}')}{\F(X \in \mathcal{X}')} \E_{\P}\left[ \frac{ \pi_{\F}(a)}{\F(a)} \frac{\pi_{\P}(a')}{\pi_{\F}(a')} \left\{ \mu_{\P} - \mu_{\F} \right\} + \frac{ \pi_{\P}(a)}{\F(a)} \left\{ \mu_{\F} - \theta_{\P} \right\} \right] \\
        &= \frac{\P(X \in \mathcal{X}')}{\F(X \in \mathcal{X}')} \E_{\P}\left[ \frac{ \pi_{\F}(a)}{\F(a)} \left\{ \frac{\pi_{\P}(a')}{\pi_{\F}(a')} - \frac{\pi_{\P}(a)}{\pi_{\F}(a)} \right\} \left\{ \mu_{\P} - \mu_{\F} \right\} + \frac{ \pi_{\P}(a)}{\F(a)}\left\{ \mu_{\P} - \mu_{\F} \right\} + \frac{ \pi_{\P}(a)}{\F(a)} \left\{ \mu_{\F} - \theta_{\P} \right\} \right] \\
        &= \frac{\P(X \in \mathcal{X}')}{\F(X \in \mathcal{X}')} \E_{\P}\left[ \frac{ \pi_{\F}(a)}{\F(a)} \left\{ \frac{\pi_{\P}(a')}{\pi_{\F}(a')} - \frac{\pi_{\P}(a)}{\pi_{\F}(a)} \right\} \left\{ \mu_{\P} - \mu_{\F} \right\} \right]. \label{eq:second-order-remainder-main}
    \end{align}
    Moreover, the difference $\textcolor{purple}{\E_{\P}\left\{ \D(\F)(Z) \right\}}  - \textcolor{teal}{\E_{\P}\left\{ \D_{\P}(\F)(Z) \right\}}$ simplifies to
    \begin{align}
        \textcolor{purple}{\E_{\P}\left\{ \D(\F)(Z) \right\}}  - \textcolor{teal}{\E_{\P}\left\{ \D_{\P}(\F)(Z) \right\}} &= \E_{\P}\left[ \frac{\1(X \in \mathcal{X}')}{\F(X \in \mathcal{X}')}  \frac{\1(A = a)}{\F(a)}\left\{ \theta_{\F} - \theta_{\P} \right\} \right] \\
        &= \E_{\P}\left[ \left\{ \frac{\P(X \in \mathcal{X}')}{\F(X \in \mathcal{X}')} \frac{\P(a)}{\F(a)} - 1\right\} \left\{\theta_{\F} - \theta_{\P}\right\} \right] + \theta_{\P} - \theta_{\F}. \label{eq:second-order-remainder-leftover}
    \end{align}
    Plugging \eqref{eq:second-order-remainder-main} and \eqref{eq:second-order-remainder-leftover} into \eqref{eq:second-order-remainder-long} yields the stated result. 
\end{proof}
To add some intuition as to how we arrived at the expansion in Theorem~\ref{thm:vonmises}, note that when $\mathcal{A} = \{0, 1\}$ and $\mathcal{X}' = \mathcal{X}$, the parameter $\theta_{\P}(1, 0) = \E[\mu_{\P}(0, X) \mid A = 1]$ is the standard counterfactual mean involved in the Average Treatment Effect on the Treated (ATT) parameter. The EIF for this counterfactual mean is well-known, and was given in \citealt{hubbard2011} as
\begin{align}
    Z \mapsto \frac{\1(A = 0)}{\P(A = 1)} \frac{\pi_{\P}(1, X)}{\pi_{\P}(0, X)} \left\{ Y - \mu_{\P}(A, X) \right\} + \frac{\1(A = 1)}{\P(A = 1)} \left\{ \mu_{\P}(0, X) - \theta_{\P}(1, 0) \right\}.
\end{align}
We reasoned by analogy to posit a candidate EIF for the parameter $\theta(a, a')$ in the general case, and verified it was indeed the correct EIF by showing the parameter and putative EIF yields a von Mises expansion with second-order remainder term $\R$. This result is formalized in Theorem~\ref{thm:vonmises}. 

\subsection{Targeted estimator}
\label{appendix:tmle}
The proposed TMLE algorithm for estimating $\psi_{a}^{\mathsf{rand}}$ is given below. Let $j \in \mathcal{J}$ index cross-fitting folds. For each fold, let $\mathcal{T}_j \subset \{1, \dots, n \}$ index the observations in the training set, and $\mathcal{V}_j \subset \{1, \dots, n \}$ index the observations in the validation set. For any $i \in \{1, \dots, n \}$, let $v[i]$ be the fold for which observation $i$ was in the validation set.

While the full procedure may appear complex, note that in the absence of cross-fitting the procedure amounts to a series of update steps on the estimated outcome regression $\mu$, all of which can be performed via standard logistic regression software. In this respect, our procedure is very similar to TMLE for parameters like the Average Treatment Effect (ATE). Due to the similarities between the replacement parameter and the traditional Average Treatment Effect on the Treated (ATT) parameter, we note that a TMLE algorithm for the ATT was proposed in \citealt{hubbard2011}. However our TMLE, which fluctuates only $\mu$, is simpler than theirs, which requires fluctuating both $\mu$ and $\pi$. 

\medskip

\textbf{Algorithm:} TMLE for random replacement parameter $\psi_a^{\mathsf{rand}}$.
\begin{enumerate}
    \item \textbf{Initialization.} For each fold $j \in \mathcal{J}$:
    \begin{enumerate}
        \item Construct estimators $\hat{\mu}_j$, $\hat{\pi}$, and $\hat{\P}_j$ using training observations $\mathcal{T}_j$. 
    \end{enumerate}
    \item For each focal player $a \in \mathcal{A}$:
    \begin{enumerate}
        \item \textbf{Targeting step.} For each $a' \in \mathcal{A}$ and cross-fitting fold $j \in \mathcal{J}$:
        \begin{enumerate}
            \item Construct a parametric fluctuation submodel parameterized by $\epsilon \in \mathbb{R}$: 
            \begin{align}
                \hat{\mu}_{j, \epsilon}(A, X) = \mathrm{logit}^{-1}\left\{ \mathrm{logit}(\hat{\mu}_j(A, X)) + \epsilon \times \hat{C}_j(A, X) \right\},
            \end{align}
            where the so-called ``clever covariate'' $\hat{C}_j$ is given by
            \begin{align}
                \hat{C}_j(A, X) = \frac{\1(X \in \mathcal{X}')}{\hat{\P}_j(X \in \mathcal{X}')} \frac{\1(A = a')}{\hat{\P}_j(a)} \frac{\hat{\pi}_j(a \mid X)}{\hat{\pi}_j(a' \mid X)}.
            \end{align}
            Note the construction of $\hat{C}$ implies that the estimates $\hat{\mu}(A, X)$ are only fluctuated when $A = a'$. 
            \item Estimate $\hat{\epsilon}$ by minimizing the empirical logistic log loss in the training dataset:
            \begin{align}
                \hat{\epsilon} = \argmin_{\epsilon \in \mathbb{R}} \left\{ -\sum_{i \in \mathcal{T}_j} \left( Y_i \log \hat{\mu}_\epsilon(A_i, X_i) + (1 - Y_i) \log\left( \hat{\mu}_\epsilon(A_i, X_i) \right) \right) \right\}.
            \end{align}
            \item Form fluctuated estimates in the validation data: for $i \in \mathcal{V}_j$, set
            \begin{align}
                \hat{\mu}^*(a', X_i) = \mathrm{logit}^{-1}\left\{ \mathrm{logit}(\hat{\mu}_{v[i]}(a', X_i)) + \hat{\epsilon} \times  \frac{\1(X \in \mathcal{X}')}{\hat{\P}_{v[i]}(X \in \mathcal{X}')}  \frac{1}{\hat{\P}_{v[i]}(a)} \frac{\hat{\pi}_{v[i]}(a \mid X_i)}{\hat{\pi}_{v[i]}(a' \mid X_i)} \right\}.
            \end{align}
        \end{enumerate}
        Practically speaking, the targeting step can be performed by estimating a logistic regression model with offset $\hat{\mu}_\epsilon(A, X)$ and single covariate $\hat{C}(A, X)$, and then making predictions from the model with the covariate set to $\hat{C}(a', X)$. 
        \item \textbf{Substitution step.} 
        Let $\mathcal{C} = \{ i : A_i = a, X_i \in \mathcal{X}' \}$. The targeted estimator is given by
        \begin{align}
            \hat{\psi}_a^{\mathsf{rand}} = \frac{1}{|\mathcal{C}|} \sum_{i \in \mathcal{C}} \frac{1}{m} \sum_{a' = 1}^m \hat{\mu}^*(a', X_i),
        \end{align}
        Note that, by construction, all $\hat{\mu}^*$ are formed using out-of-sample nuisance predictions from the cross-fitting procedure.
    \end{enumerate}
    \item \textbf{Output.} Return $(\hat{\psi}_{1}^{\mathsf{rand}}, \hat{\psi}_{m}^{\mathsf{rand}})^\top$, the vector of targeted estimates.
\end{enumerate}

\paragraph{Proof of Theorem~\ref{thm:consistency}}. 
We give a sketch of the proof for that the consistency and asymptotic normality cross-fitted TMLE is well-established \citep{zheng2011cvtmle}, and similar results are available for TMLE estimators of direct and indirect standardization parameters \citep{susmann2025providerevaluation}. In the following we use the shorthand notation $\P[f] = \int f(z) d\P(z)$ to denote the expectation of a function $f$ with respect to $\P$, and $\P_n[f] = n^{-1}\sum_{i=1}^n f(Z_i)$ to denote the empirical mean of $f$. We slightly abuse notation to write $\D_{\P_n^*}$ to denote the estimate of the EIF based on the updated estimate $\P_n^*$ output from the TMLE algorithm.

First, note that empirical mean estimators of $\hat{\P}(X \in \mathcal{X}')$ and $\hat{\P}(a)$ (for all $a \in \mathcal{A}$) are consistent by the large law of numbers. Therefore, consistency and asymptotic normality depend principally on the properties of the nuisance estimators for the functionals $\mu$ and $\pi$. The core of the proof lies in the decomposition
\begin{align}
    \hat{\psi}_a^{\mathsf{rand}} - \psi_a^{\mathsf{rand}}(\P) =& -\P_n\left[ \D_{\P^*} \right] \\
    & + (\P_n - \P)\left[ \D_{\P} \right] \\
    & + (\P_n - \P)\left[ \D_{\P^*} - \D_{\P} \right] \\
    & + \R(\P, \P^*).
\end{align}
The first line is the empirical mean of the EIF evaluated on the TMLE updated distribution $\P^*$, which, by construction of the TMLE updates, is $o_{\P}(n^{-1/2})$. The third line, referred to as the \textit{empirical process term}, can be controlled either through complexity restrictions (e.g. Donsker class conditions) on the nuisance estimates, or through the use of cross-fitting. For the fourth line, the second-order remainder term, recall that by assumption $\| \hat{\mu} - \mu \| = o_{\P}(n^{-1/4})$ and $\| \hat{\pi} - \pi \| = o_{\P}(n^{-1/4})$, implying that$\| \hat{\mu} - \mu \| \| \hat{\pi} - \pi \| = o_{\P}(n^{-1/2})$. Using this with the Cauchy-Schwartz inequality implies the first term in $\R$ is $o_{\P}(n^{-1/2})$; the second term of $\R$ is $o_{\P}(n^{-1/2})$ by the law of large numbers. 
\end{appendix}

\section{Additional results}
\renewcommand\thefigure{\thesection\arabic{figure}}    
\setcounter{figure}{0}    
\begin{figure}[ht]
    \centering
    \includegraphics[width=0.75\linewidth]{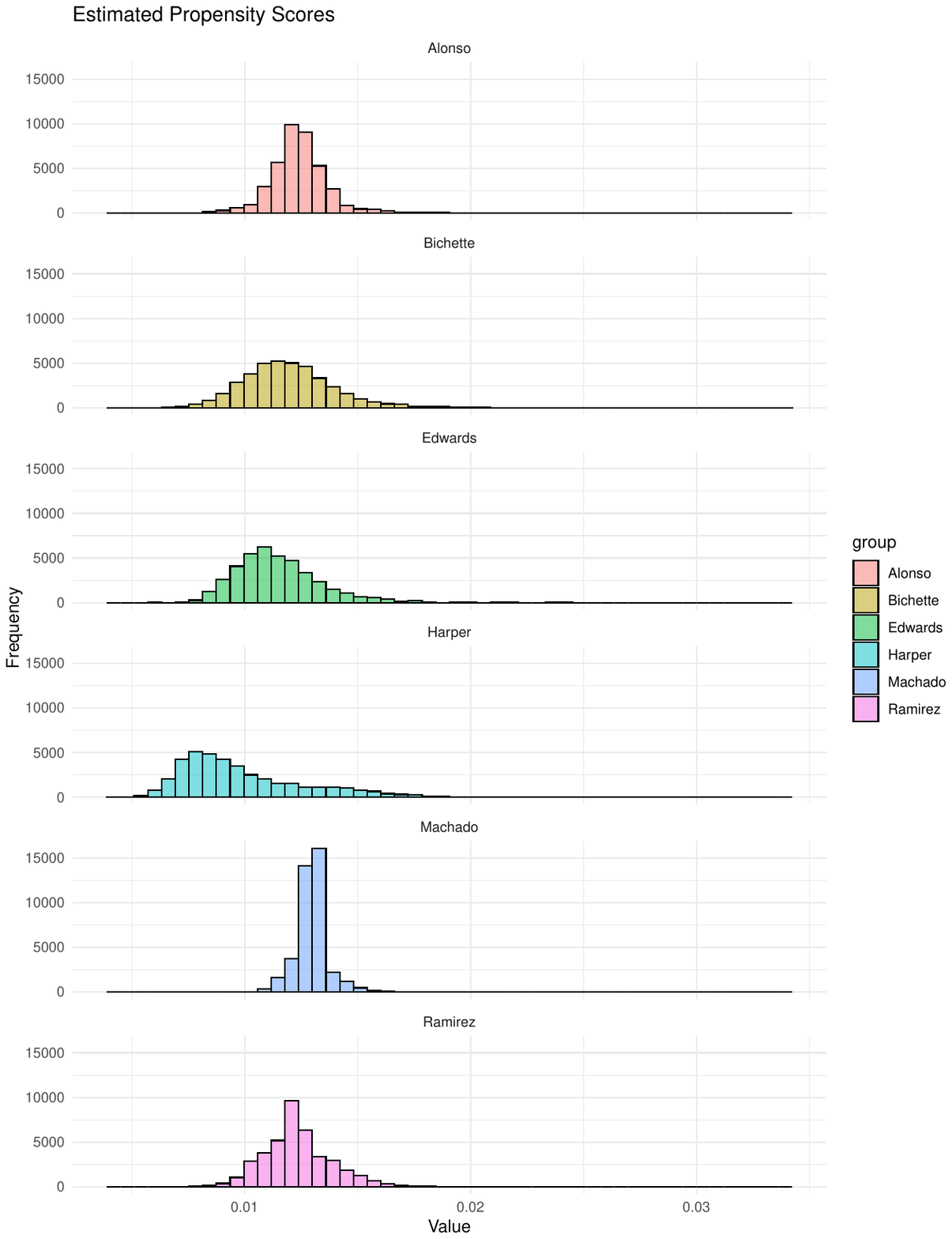}
    \caption{\small Comparison of the distribution of estimated propensity scores for the indirect model for Bryce Harper, Pete Alonso, Jose Ramirez, Bo Bichette, Xavier Edwards and Manny Machado.}
    \label{fig:mlb-harper}
\end{figure}

\end{document}